%% file: main.tex
\documentclass[12pt]{article}

\usepackage[british]{babel}

\usepackage{natbib} 
\usepackage{mathtools,amssymb} 
\usepackage{amsthm}
\usepackage{booktabs} 
\usepackage[hidelinks=true]{hyperref}
\usepackage{cleveref}
\usepackage{mhchem}
\usepackage{algorithm}
\usepackage[noend]{algpseudocode}

\usepackage[margin=2cm]{geometry}
\usepackage{setspace}
\usepackage{parskip}

\usepackage{authblk}

\newcommand{\ABC}{\mathrm{ABC}}
\newcommand{\hi}{\mathrm{hi}}
\newcommand{\lo}{\mathrm{lo}}
\newcommand{\mf}{\mathrm{mf}}

\newtheorem{thm}{Theorem}
\newtheorem{cor}[thm]{Corollary}
\newtheorem{defn}[thm]{Definition}
\newtheorem{lem}[thm]{Lemma}
\newtheorem{prop}[thm]{Proposition}

\title{Efficient Multifidelity Likelihood-Free Bayesian Inference with Adaptive Computational Resource Allocation}
\author[1]{Thomas P Prescott} 
\author[2]{David J Warne} 
\author[3]{Ruth E Baker}
\affil[1]{Alan Turing Institute, London NW1 2DB, United Kingdom}
\affil[2]{QUT Centre for Data Science, Queensland University of Technology, Brisbane, QLD 4000, Australia}
\affil[3]{Wolfson Centre for Mathematical Biology, Mathematical Institute, University of Oxford, Oxford OX2 6GG, United Kingdom}

\begin{document}
\maketitle

\onehalfspacing

\begin{abstract}
Likelihood-free Bayesian inference algorithms are popular methods for calibrating the parameters of complex, stochastic models, required when the likelihood of the observed data is intractable.
These algorithms characteristically rely heavily on repeated model simulations.
However, whenever the computational cost of simulation is even moderately expensive, the significant burden incurred by likelihood-free algorithms leaves them unviable in many practical applications.
The multifidelity approach has been introduced (originally in the context of approximate Bayesian computation) to reduce the simulation burden of likelihood-free inference without loss of accuracy, by using the information provided by simulating computationally cheap, approximate models in place of the model of interest.
The first contribution of this work is to demonstrate that multifidelity techniques can be applied in the general likelihood-free Bayesian inference setting.
Analytical results on the optimal allocation of computational resources to simulations at different levels of fidelity are derived, and subsequently implemented practically.
We provide an adaptive multifidelity likelihood-free inference algorithm that learns the relationships between models at different fidelities and adapts resource allocation accordingly, and demonstrate that this algorithm produces posterior estimates with near-optimal efficiency. 
\end{abstract}

\section{Introduction}\label{s:intro}
Across domains in engineering and science, parametrised mathematical models are often too complex to analyse directly.
Instead, many \emph{outer-loop applications} \citep{Peherstorfer2018}, such as model calibration, optimization, and uncertainty quantification, rely on repeated simulation to understand the relationship between model parameters and behaviour.
In time-sensitive and cost-aware applications, the typical computational burden of such simulation-based methods makes them impractical. 
Multifidelity methods, reviewed by \citet{Peherstorfer2016,Peherstorfer2018}, are a family of approaches that exploit information gathered from simulations, not only of a single model of interest, but also of additional approximate or surrogate models.
In this article, the term \emph{model} refers to the underlying mathematical abstraction of a system in combination with the computer code used to implement simulations.
Thus, `model approximation' may refer to mathematical simplifications and/or approximations in numerical methods.
The fundamental challenge when implementing multifidelity techniques is the allocation of computational resources between different models, for the purposes of balancing a characteristic trade-off between maintaining accuracy and saving computational burden.

In this work, we consider a specific outer-loop application that arises in Bayesian statistics, the goal of which is to calibrate a parametrised model against observed data.
Bayesian inference uses the likelihood of the observed data to update a prior distribution on the model parameters into a posterior distribution, according to Bayes's rule.
In the situation where the likelihood of the data cannot be calculated, we rely on so-called likelihood-free methods that provide estimates of the likelihood by comparing model simulations to data.
For example, approximate Bayesian computation (ABC) is a widely-known likelihood-free inference technique \citep{ABC2020,Sunnaaker2013},
where the likelihood is typically estimated as a binary value, recording whether or not the distance between a simulation and the observed data falls within a given threshold.
Other likelihood-free methods are also available, such as pseudo-marginal methods and Bayesian synthetic likelihoods (BSL).
In this work, we develop a generalised likelihood-free framework for which ABC, pseudo-marginal and BSL can be expressed as specific cases, as described in \Cref{s:likelihood-free}.

The significant cost of likelihood-free inference has motivated several successful proposals for improving the efficiency of likelihood-free samplers, such as (in the specific context of ABC) ABC-MCMC~\citep{Marjoram2003} and ABC-SMC~\citep{Sisson2007,Toni2009,Moral2011}.
These approaches aim to efficiently explore parameter space by avoiding the proposal of low-likelihood parameters, reducing the required number of expensive simulations required and reducing the ABC rejection rate.
However, an `orthogonal' technique for improving the efficiency of likelihood-free inference is to instead ensure that each simulation-based likelihood estimate is, on average, less computationally expensive to generate.

In previous work, \citet{Prescott2020,Prescott2021} investigated multifidelity approaches to likelihood-free Bayesian inference \citep{Cranmer2020}, with a specific focus on ABC~\citep{ABC2020,Sunnaaker2013}.
Suppose that there exists a low-fidelity approximation to the parametrised model of interest, and that the approximation is relatively cheap to simulate.
Monte Carlo estimates of the posterior distribution, with respect to the likelihood of the original high-fidelity model, can be constructed using the simulation outputs of the low-fidelity approximation.
\citet{Prescott2020} showed that using the low-fidelity approximation introduces no further bias, so long as, for any parameter proposal, there is a positive probability of simulating the high-fidelity model to check and potentially correct a low-fidelity likelihood estimate.
The key to the success of the multifidelity ABC (MF-ABC) approach is to choose this positive probability to be suitably small, thereby simulating the original model as little as possible, while ensuring it is large enough that the variance of the resulting Monte Carlo estimate is suitably small.
The result of the multifidelity approach is to reduce the expected cost of estimating the likelihood for each parameter proposal in any Monte Carlo sampling algorithm.
In subsequent work, \citet{Prescott2021} showed that this approach integrates with sequential Monte Carlo (SMC) sampling for efficient parameter space exploration \citep{Toni2009,Moral2011,Drovandi2011}.
Thus, the synergistic effect of combining multifidelity and with SMC to improve the efficiency of ABC has been demonstrated.

Multifidelity ABC can be compared with previous techniques for exploiting model approximation in ABC, such as Preconditioning ABC \citep{Warne2021:JCGS}, Lazy ABC (LZ-ABC) \citep{Prangle2016}, and Delayed Acceptance ABC (DA-ABC) \citep{Christen2005,Everitt2021}.
Similarly to sampling techniques such as SMC, the preconditioning approach seeks to explore parameter space more efficiently, by proposing parameters for high-fidelity simulation with greater low-fidelity posterior mass.
In contrast, each of MF-ABC, LZ-ABC, and DA-ABC seeks to make each parameter proposal quicker to evaluate, on average, by using the output of the low-fidelity simulation to directly decide whether to simulate the high-fidelity model.
In both LZ-ABC and DA-ABC, a parameter proposal is either (a) rejected early, based on the simulated output of the low-fidelity model, or (b) sent to a high-fidelity simulation, to make a final decision on ABC acceptance or rejection.
The distinctive aspect of MF-ABC is that step (a) is different; it is not necessary to reject early to avoid high-fidelity simulation.
Instead the low-fidelity simulation can be used to make the accept/reject decision directly.
In both DA-ABC and MF-ABC, the decision between (a) or (b) is based solely on whether the low-fidelity simulation would be accepted or rejected.
In contrast, LZ-ABC allows for a much more generic decision of whether to simulate the high-fidelity model, requiring an extensive exploration of practical tuning methods.

In this paper, we will show that the multifidelity approach can be applied to any simulation-based likelihood-free inference methodology, including but not limited to ABC.
We achieve this by developing a generalised framework for likelihood-free inference, and deriving a multifidelity method to operate in this framework.
A successful multifidelity likelihood-free inference algorithm requires us to determine how many simulations of the high-fidelity model to perform, based on the parameter value and the simulated output of the low-fidelity model.
We provide theoretical results and practical, automated tuning methods to allocate computational resources between two models, designed to optimise the performance of multifidelity likelihood-free importance sampling.

\subsection{Outline}
In \Cref{s:LikelihoodFree} we review likelihood-free Bayesian inference through constructing a generalised likelihood-free framework.
\Cref{s:mf} shows how the theory underpinning MF-ABC, as introduced by \citet{Prescott2020}, can be applied in this general likelihood-free Bayesian context.
We analyse the performance of the resulting multifidelity likelihood-free importance sampling algorithm.
The main results of this paper are set out in \Cref{s:performance}, in which we determine the optimal allocation of computational resources between the two models to achieve the best possible performance of multifidelity inference.
\Cref{s:mfimplementation} explores how to practically allocate computation between model fidelities, by adaptively evolving the allocation in response to learned relationships between simulations at each fidelity across parameter space.
We illustrate adaptive multifidelity inference by applying the algorithm to a simple biochemical network motif in \Cref{s:eg}.
We show that, using a low-fidelity Michaelis--Menten approximation together with the exact model (both simulated using the exact algorithm of \citet{Gillespie1977}) our adaptive implementation of multifidelity likelihood-free inference can achieve a quantifiable speed-up in constructing posterior estimates to a specified variance and with no additional bias.
Code for this example, developed in Julia~1.6.2~\citep{JULIA}, is available at \href{https://github.com/tpprescott/mf-lf}{\texttt{github.com/tpprescott/mf-lf}}.
Finally, in \Cref{s:end} we discuss how greater improvements may be achieved for more challenging inference tasks.

\section{Likelihood-free inference}
\label{s:LikelihoodFree}

We consider a stochastic model of the data generating process, defined by a distribution with parametrised probability density function, $f(\cdot \mid \theta)$, where the parameter vector $\theta$ takes values in a parameter space $\Theta$.
For any $\theta$, the model induces a probability density, denoted $f(y \mid \theta)$, on observable outputs, with $y$ taking values in an output space $\mathcal Y$.
We note that the model is usually implemented in computer code to allow simulation, through which outputs $y \in \mathcal Y$ can be generated.
We write $y \sim f(\cdot\mid\theta)$ to denote simulation of the model $f$ given parameter values $\theta$.
Taking the experimentally observed data $y_0 \in \mathcal Y$, we define the likelihood function to be a function of $\theta$ using the density, $\mathcal L(\theta) = f(y_0 \mid \theta)$, of the observed data under this model.

Bayesian inference updates prior knowledge of the parameter values, $\theta \in \Theta$, which we encode in a prior distribution with density $\pi(\theta)$.
The information provided by the experimental data, encoded in the likelihood function, $\mathcal L(\theta)$, is combined with the prior using Bayes' rule to form a posterior distribution, with density
\begin{equation}
\label{eq:posterior}
\pi(\theta \mid y_0) = \frac{\mathcal L(\theta) \pi(\theta)}{Z},
\end{equation}
where $Z = \int \mathcal L(\theta) \pi(\theta) ~\mathrm d\theta$ normalises $\pi(\cdot \mid y_0)$ to be a probability distribution on $\Theta$.
For a given, arbitrary, integrable function $G : \Theta \rightarrow \mathbb R$, we take the goal of the inference task as the production of a Monte Carlo estimate of the posterior expectation, 
\begin{equation}
\label{eq:Phibar} 
\bar G = \mathbf E(G \mid y_0) = \int_\Theta G(\theta) \pi(\theta \mid y_0)~\mathrm d\theta,
\end{equation}
conditioned on the observed data.

\subsection{Approximating the likelihood with simulation}\label{s:likelihood-free}
In most practical settings, models tend to be sufficiently complicated that calculating $\mathcal L(\theta) = f(y_0 \mid \theta)$ for $\theta \in \Theta$ is intractable.
In this case, we exploit the ability to produce independent simulations from the model,
\begin{subequations}
\label{eq:ell}
\begin{align}
\mathbf y &= (y_1, \dots, y_K), \\
y_k &\sim f(\cdot \mid \theta).
\end{align}
In the following, we will slightly abuse notation by using the shorthand $\mathbf y \sim f(\cdot \mid \theta)$ to represent $K$ independent, identically distributed draws from the parametrised distribution $f(\cdot\mid\theta)$.

Given the observed data, $y_0$, we can define a real-valued function, referred to as a \emph{likelihood-free weighting},
\begin{equation}
\omega : (\theta, \mathbf y) \mapsto \mathbb R,
\end{equation}
\end{subequations}
which varies over the joint space of parameter values and simulation outputs.
Here, $\omega$ is a function of a parameter value, $\theta$, and a vector, $\mathbf y$, of stochastic simulations.
For a fixed $\theta$, we can take conditional expectations of $\omega(\theta, \mathbf y)$ over the probability density of simulations, $f(\mathbf y\mid\theta)$, to define an \emph{approximate likelihood function},
\begin{subequations}
\label{eq:lhf}
\begin{equation}
L_\omega(\theta) = \mathbf E(\omega \mid \theta) = \int \omega(\theta, \mathbf y) f(\mathbf y \mid \theta) ~\mathrm d\mathbf y,
\end{equation}
where $L_\omega(\theta)$ is assumed to be approximately equal to the modelled likelihood function, $\mathcal L(\theta)$, up to a constant of proportionality.
Note that, given $\theta$, the random value of the likelihood-free weighting, $\omega(\theta, \mathbf y)$, determined by $K$ stochastic simulations, $\mathbf y \sim f(\cdot\mid\theta)$, is a Monte Carlo estimate of the approximate likelihood function, $L_\omega(\theta)$.

The approximate likelihood function is used to define the \emph{likelihood-free approximation to the posterior},
\begin{equation}
\label{eq:posterior_approximation}
\pi_\omega(\theta \mid y_0) = \frac{L_\omega(\theta) \pi(\theta)}{Z_\omega},
\end{equation}
where the normalisation constant $Z_\omega = \int L_\omega(\theta) \pi(\theta)~\mathrm d\theta$ ensures that $\pi_\omega$ is a probability distribution.
The likelihood-free approximation to the posterior, $\pi_\omega(\theta \mid y_0)$, subsequently induces a biased estimate of $\bar G$, given by
\begin{equation}
\bar G_\omega = \mathbf E_{\pi_\omega}(G \mid y_0) = \int G(\theta) \pi_\omega(\theta \mid y_0) ~\mathrm d\theta.
\end{equation}
\end{subequations}
In this situation, the success of likelihood-free inference depends on ensuring that the likelihood-free weighting, $\omega(\theta, \mathbf y)$ is chosen such that the squared difference $(\bar G - \bar G_\omega)^2$ between the posterior expectation, $\bar G$, and its likelihood-free approximation, $\bar G_\omega$, is as small as possible.

\subsubsection{Example: ABC}
Approximate Bayesian computation (ABC) is a widely-used example of likelihood-free inference, where
\[
\omega_\ABC(\theta, \mathbf y) = \frac{1}{K} \sum_{k=1}^K \mathbf I(d(y_k, y_0) \leq \epsilon),
\]
is the random fraction of $K$ simulations, $y_k$, which are within a distance of $\epsilon$ of the observed data, as measured by the metric $d$.
\citet{Sisson2007,ABC2020} note that it is standard practice to choose $K=1$.
Taking the conditional expectation of $\omega_\ABC$, given $\theta$, this likelihood-free weighting induces the ABC approximation to the likelihood, \[ L_\ABC (\theta) = \mathbf E(\omega_\ABC \mid \theta) = \mathbf P(d(y, y_0) \leq \epsilon \mid \theta),\] for any $K$.
Under appropriate choices of $d$ and $\epsilon$, the approximate likelihood function, $L_\ABC(\theta)$, may be considered approximately proportional to the likelihood, $\mathcal L(\theta)$.

\subsubsection{Example: Bayesian synthetic likelihood}
The simplest implementation of the Bayesian synthetic likelihood approach replaces the true likelihood with a Monte Carlo likelihood-free weighting based on a Gaussian density,
\[
\omega_{\mathrm{SL}}(\theta, \mathbf y)
= 
\mathcal N(y_0 ; \mu(\mathbf y), \Sigma^2(\mathbf y)),
\]
with random mean,
$\mu = \sum_k y_k / K$,
and covariance,
$\Sigma^2 = \sum_k (y_k - \mu)(y_k - \mu)^T / K$,
given by the empirical mean and covariance of $K$ simulations.
Taking conditional expectations of $\omega_{\mathrm{SL}}$, given $\theta$, induces the Bayesian synthetic likelihood, $L_{\mathrm{SL}}(\theta) = \mathbf E(\omega_{\mathrm{SL}}(\theta, \mathbf y) \mid\theta)$, as an approximation of the likelihood, $\mathcal L(\theta)$.

\subsubsection{Example: Pseudo-marginal method}
For the pseudo-marginal approach, introduced by \citet{Andrieu2009}, we suppose that there exists a simulation-based estimator, $\omega(\theta, \mathbf y)$, such that the conditional expectation $L_\omega(\theta) = \mathbf E(\omega \mid \theta)$ is an \emph{unbiased} estimate of $\mathcal L(\theta) = f(y_0 \mid \theta)$.
For example, following \citet{Warne2020}, suppose that an intractable density $f(y\mid\theta)$ arising from a stochastic model can be decomposed into an underlying latent model, $x \sim g(\cdot\mid\theta)$, and an observation model, $y \sim h(\cdot\mid\theta,x)$, such that $f(y\mid\theta) = \int h(y\mid\theta,x) g(x\mid\theta) \mathrm dx$.
Assume that the probability densities $h(y\mid\theta,x)$ of the observation model can be calculated.
Then, for simulations $x_k \sim g(\cdot\mid\theta)$ of the latent model, where $k=1,\dots,K$, we can write
\[
\omega(\theta, \mathbf x) = \frac{1}{K} \sum_{k} h(y_0 \mid \theta, x_k),
\]
as a likelihood-free weighting.
Taking expectations over $\mathbf x \sim g(\cdot\mid\theta)$, we have $L_\omega(\theta) = \mathbf E(\omega \mid \theta) = f(y_0 \mid \theta) = \mathcal L(\theta)$. 
Thus, $L_\omega(\theta)$ is an \emph{exact approximation}~\citep{Drovandi2019}.

\subsection{Likelihood-free importance sampling}
\label{s:is}
A simple approach to estimating the likelihood-free approximate posterior mean, $\bar G_\omega$, is to use importance sampling.
We assume that parameter proposals $\theta_i \sim q(\cdot)$, for $i=1,\dots,N$, can be sampled from a given importance distribution, the support of which must include the prior support.
In practice, we need only know importance density values, $q(\theta)$, up to a multiplicative constant, but for simplicity we assume that $q(\theta)$ is known.
We also assume that we have access to the prior probability density, $\pi(\theta)$.

The \emph{likelihood-free importance sampling} algorithm is described in \Cref{alg:is}. 
This algorithm requires the specification of an importance distribution, $q$, and a likelihood-free weighting, $\omega(\theta, \mathbf y)$, with conditional expectation, $L_\omega(\theta) = \mathbf E(\omega\mid\theta)$. 
The output of \Cref{alg:is}, $\hat G$, is an estimate of the likelihood-free approximate posterior expectation, $\bar G_\omega = \mathbf E_{\pi_\omega}(G \mid y_0)$.
In \Cref{thm:consistent}, we prove the standard result that $\hat G$ is a \emph{consistent} estimate of $\bar G_\omega$, and quantify the dominant behaviour of the mean squared error in the limit of large sample sizes, $N \rightarrow \infty$.
For notational simplicity, we define the function $\Delta(\theta) = G(\theta) - \bar G_\omega$ to recentre $G$ at the approximate posterior mean, and denote the Monte Carlo error between $\hat G$ and $\bar G_\omega$ as the estimated mean value of $\Delta$, denoted by $\hat \Delta = \hat G - \bar G_\omega$.

\begin{algorithm}
\caption{Likelihood-free importance sampling.}
\label{alg:is}
\begin{algorithmic}
\Require Prior, $\pi$; importance distribution, $q$; likelihood-free weighting, $\omega$; model $f(\cdot \mid \theta)$; stop condition, \texttt{stop}; target function, $G$.
\Statex
\State Set counter $i=0$.
\Repeat 
	\State Increment counter $i \gets i+1$;
	\State Sample $\theta_i \sim q(\cdot)$;
	\State Simulate $\mathbf y_i \sim f(\cdot \mid \theta_i)$;
	\State Calculate weight, 
		\begin{equation}
		\label{eq:w}
			w_i = w(\theta_i, \mathbf y_i) = \frac{\pi(\theta_i)}{q(\theta_i)} \omega(\theta_i, \mathbf y_i);
		\end{equation}
\Until{\texttt{stop} = \texttt{true}}
\State \Return Weighted sum,
\[
\hat G = \sum_{i=1}^N w_i G(\theta_i) \bigg/ \sum_{j=1}^N w_j.
\]
\end{algorithmic}
\end{algorithm}

\begin{thm}
\label{thm:consistent}
For the weighted sample values $(\theta_i, w_i)$ produced in each iteration of \Cref{alg:is}, let $w$ denote the random value of the weight $w_i$, and let $\Delta$ denote the random value of $\Delta(\theta_i)$.
The mean squared error (MSE) of the output, $\hat G$, of \Cref{alg:is}, as an estimator for the approximate posterior expectation, $\bar G_\omega$, is given to leading order by
\begin{equation}
\label{eq:MSE}
\mathbf E \left( \hat \Delta^2 \right) = \left[ \frac{\mathbf E \left(w^2 \Delta^2\right)}{\mathbf E(w)^2} \right] \frac{1}{N} + O \left(\frac{1}{N^2}\right).
\end{equation}
Thus, $\hat G$ is a consistent estimator of $\bar G_\omega$.
\end{thm}
\begin{proof}
The Monte Carlo estimate produced by \Cref{alg:is}, $\hat G = R/S$, is the ratio of two random variables: the weighted sum, $R=\sum_{i=1}^N w(\theta_i,\mathbf y_i) G(\theta_i)$, and the normalising sum, $S = \sum_{i=1}^N w(\theta_i, \mathbf y_i)$.
We write the function $\Phi(r,s) = (r/s - \bar G_\omega)^2$, and note that the MSE is the expected value of the function $\hat \Delta^2 = \Phi(R,S)$.
Using the delta method, we take expectations of the second-order Taylor expansion of $\Phi(R,S)$ about $(\mu_R, \mu_S) = (\mathbf E(R), \mathbf E(S))$, to give
\begin{align*}
\mathbf E\left( \hat \Delta^2 \right) = \mathbf E\left( \Phi(R,S) \right)
&= \Phi(\mu_R, \mu_S) \\
&\quad + \frac{1}{\mu_S^2} \left[ \mathrm{Var}(R) + \left( 2 \bar G_\omega - \frac{4\mu_R}{\mu_S} \right) \mathrm{Cov}(R,S) + \left( \frac{3 \mu_R - 2 \bar G_\omega \mu_S}{\mu_S^2} \right) \mu_R  \mathrm{Var}(S) \right] \\ 
&\quad + O \left( \frac{\mathbf E \left( \left( \left( R-\mu_R \right)+ \left( S-\mu_S \right) \right)^3 \right)}{\mu_s^3} \right).
\end{align*}
Taking expectations with respect to $N$ independent draws of $(\theta, \mathbf y)$ with density $f(\mathbf y \mid \theta) q(\theta)$, it is straightforward to write
\begin{align*}
\mu_R = \mathbf E(R) &= N \mathbf E(w G) = N Z_\omega \bar G_\omega, \\
\mu_S = \mathbf E(S) &= N \mathbf E(w) = N Z_\omega,
\end{align*}
where we recall that $\mathbf E(w) = Z_\omega = \int L_\omega(\theta) \pi(\theta)~\mathrm d\theta$ is the normalising constant in \Cref{eq:posterior_approximation}.
We substitute these expectations into the Taylor expansion of $\mathbf E(\hat \Delta^2)$, noting that the leading-order term, $\Phi(\mu_R, \mu_S)$, is zero.
Thus, we can write the dominant behaviour of the MSE as
\begin{align*}
\mathbf E\left(\hat \Delta^2 \right) &= \frac{1}{N^2 Z_\omega^2} \left[ \mathrm{Var}(R) - 2\bar G_\omega \mathrm{Cov}(R,S) + \bar G_\omega^2 \mathrm{Var}(S) \right] + O\left( \frac{\mathbf E \left( \left( \left( R-\mu_R \right)+ \left( S-\mu_S \right) \right)^3 \right)}{N^3} \right) \\
&= \frac{1}{N^2 Z_\omega^2} \mathrm{Var}(R - \bar G_\omega S) + O\left( \frac{\mathbf E \left( \left( \left( R-\mu_R \right)+ \left( S-\mu_S \right) \right)^3 \right)}{N^3} \right),
\end{align*}
as $N\rightarrow\infty$.
Substituting into this expression the definitions of $R$ and $S$ as summations of $N$ independent identically distributed random variables, we have
\[
\mathbf E(\hat \Delta^2) = \frac{1}{N^2 Z_\omega^2} N \mathrm{Var}(w\Delta) + O\left( \frac{N}{N^3} \right) ,
\]
and \Cref{eq:MSE} follows, on noting that $\mathbf E(w\Delta) = 0$ and that $\mathbf E(w) = Z_\omega$.
\end{proof}

\Cref{thm:consistent} determines the leading-order behaviour of the MSE of the output of \Cref{alg:is} in terms of sample size.
We can also quantify the performance of this algorithm in terms of how the MSE decreases with increasing the overall computational budget.

\begin{cor}
\label{cor:performance}
Let the computational cost of each iteration of \Cref{alg:is} be denoted by the random variable~$C$.
The leading order behaviour of the MSE of $\hat G$ as an estimate of $\bar G_\omega$ is
\begin{equation}
\label{eq:performance}
\mathbf E \left( \hat \Delta^2 \right) = \left[ \frac{\mathbf E(C)\mathbf E \left(w^2 \Delta^2\right)}{\mathbf E(w)^2} \right] \frac{1}{C_{\mathrm{tot}}} + O\left(\frac{1}{C_{\mathrm{tot}}^2}\right),
\end{equation}
as the total simulation budget $C_{\mathrm{tot}} \rightarrow \infty$.
\end{cor}
\begin{proof}
As the given computational budget increases, $C_{\mathrm{tot}} \rightarrow \infty$, the Monte Carlo sample size that can be produced in that budget increases on the order of $N \sim C_{\mathrm{tot}} / \mathbf E(C)$.
On substituting this expression into \Cref{eq:MSE}, the result follows.
\end{proof}

We can use the leading-order coefficient of $1/C_{\mathrm{tot}}$ in \Cref{eq:performance} to quantify the performance of likelihood-free importance sampling.
Importantly, this expression explicitly depends on the expected computational cost, $C$, of each iteration of \Cref{alg:is}.
In the importance sampling context, the optimal importance distribution $q$ should seek to minimise this coefficient, by trading off a preference for parameter values with lower computational burden against ensuring small variability in the weighted errors, $w\Delta$.
However, for simplicity, we will assume in this paper that $q$ is fixed.
Instead, we seek ways to use model approximations to directly reduce the leading-order coefficient in \Cref{eq:performance}, based on the identified trade-off between decreasing computational burden, $C$, and controlling the variance of the weighted error, $w\Delta$.

\section{Multifidelity inference}
\label{s:mf}

In \Cref{cor:performance}, the performance of \Cref{alg:is} is quantified explicitly in terms of how the Monte Carlo error between the estimate, $\hat G$, and the approximated posterior mean, $\bar G_\omega$, decays with increasing computational budget, $C_{\mathrm{tot}}$.
It initially appears reasonable to conclude that the linear dependence of the performance on the expected iteration time, $\mathbf E(C)$, implies that if we can speed up the simulation step of \Cref{alg:is}, then we can significantly reduce the MSE for a given computational budget.

Suppose that there exists an alternative model that we can use in \Cref{alg:is} in place of the original model, $f(\cdot \mid \theta)$, such that the expected computation time for each iteration, $\mathbf E(C)$, is significantly reduced.
There are two important issues that prevent this being a viable option for improving the efficiency of likelihood-free inference.
The first problem is that we need to be able to quantify the effect of the alternative model on the ratio $\mathbf E(w^2 \Delta^2) / \mathbf E(w)^2$ to ensure that the \emph{overall} performance of the algorithm is improved.
It is not sufficient to show that the computational burden of each iteration is reduced, if too many more iterations are subsequently required to achieve a specified MSE.

The second problem arises from the observation that the limiting value of $\hat G$, as output from \Cref{alg:is}, is $\bar G_\omega$, with residual bias,
\[
\lim_{C_{\mathrm{tot}} \rightarrow \infty} \mathbf E \left( \left( \hat G - \bar G \right)^2 \right) = (\bar G_\omega - \bar G)^2 \neq 0,
\]
recalling that $\bar G_\omega$ is the approximate posterior expectation induced by $L_\omega(\theta) = \mathbf E(\omega\mid\theta)$, and where the approximand, $\bar G$, is the posterior expectation induced by the likelihood, $\mathcal L(\theta) = f(y_0 \mid \theta)$.
We will identify this limiting residual squared bias, $(\bar G_\omega - \bar G)^2$, as the \emph{fidelity} of the model/likelihood-free weighting pair.
We emphasise here that the fidelity depends both on the model and the likelihood-free weighting used in \Cref{alg:is}, and is contextual to the target function, $G$.
For a given posterior mean, $\bar G$, a model and likelihood-free weighting pair for which the value of $(\bar G_\omega - \bar G)^2$ is small is termed high-fidelity, while larger values of $(\bar G_\omega - \bar G)^2$ are termed low-fidelity.
Thus, if we use an alternative model in place of $f$ in \Cref{alg:is}, the model (and likelihood-free weighting) may be too low-fidelity, in the sense of having too large a residual squared bias versus the posterior expectation of interest, $\bar G$.

The \emph{multifidelity} framework overcomes both these problems, by removing the need for a binary choice between the expensive model of interest and its cheaper alternative.
Instead, we carry out likelihood-free inference using information from both models.
We introduce the multifidelity likelihood-free importance sampling algorithm in \Cref{s:mf_is}.
In \Cref{s:hi}, we show that multifidelity likelihood-free importance sampling loses no fidelity versus high-fidelity importance sampling.
\Cref{s:performance} contains the main analytical results of this paper, in which we explore the conditions under which multifidelity inference can improve the performance of likelihood-free importance sampling with \emph{finite} computational budgets, as quantified by the leading-order characterisation of the MSE given in \Cref{eq:performance}.

\subsection{Multifidelity likelihood-free importance sampling}
\label{s:mf_is}

We denote the high-fidelity model and likelihood-free weighting as $f_\hi$ and $\omega_\hi$, respectively.
The likelihood under the high-fidelity model is denoted $\mathcal L_\hi(\theta) = f_\hi(y_0 \mid \theta)$, and is assumed to be intractable.
Following the notation introduced in \Cref{eq:lhf}, the high-fidelity pair $f_\hi$ and $\omega_\hi$ induce the approximate likelihood, $L_\hi(\theta) = \mathbf E(\omega_\hi \mid \theta)$ and the corresponding likelihood-free approximation to the posterior expectation, $\bar G_\hi$.
We further assume that simulating each $\mathbf y_\hi \sim f_\hi(\cdot\mid\theta)$ is computationally expensive.
This computational expense motivates the use of an approximate, low-fidelity model and likelihood-free weighting, denoted $f_\lo$ and $\omega_\lo$, respectively, inducing the approximate likelihood, $L_{\lo}(\theta) = \mathbf E(\omega_\lo \mid \theta)$, and corresponding likelihood-free approximation to the posterior expectation, $\bar G_{\lo}$.
We note that the low-fidelity model, $f_\lo$, induces its own likelihood, $L_\lo = f_\lo(y_0 \mid \theta)$, but assume that this remains intractable, requiring instead the simulation-based Bayesian approach.
However, we assume that simulations of the low-fidelity model, $\mathbf y_\lo \sim f_\lo(\cdot\mid\theta)$, are significantly cheaper to produce compared to simulations of the high-fidelity model.

Given the models $f_\lo$ and $f_\hi$, we will term the joint distribution $f_\mf(\mathbf y_\lo, \mathbf y_\hi \mid\theta)$ a \emph{multifidelity model} when $f_\mf$ has marginals equal to the low- and high-fidelity densities, $f_\lo(\mathbf y_\lo \mid \theta)$ and $f_\hi(\mathbf y_\hi \mid \theta)$.
The models \emph{may} be conditionally independent, such that $f_\mf(\mathbf y_\lo, \mathbf y_\hi \mid\theta) = f_\lo(\mathbf y_\lo \mid \theta) f_\hi(\mathbf y_\hi \mid \theta)$, in which case simulations at each model fidelity can be carried out independently given $\theta$.
Furthermore, if the simulations are conditionally independent, this means that the resulting likelihood-free weights, $\omega_\lo(\theta, \mathbf y_\lo)$ and $\omega_\hi(\theta, \mathbf y_\hi)$, are also conditionally independent.

However, in the more general definition of the multifidelity model as a joint distribution, we allow for \emph{coupling} between the two fidelities.
Conditioned on the low-fidelity simulations, $\mathbf y_\lo$, and on parameter values, $\theta$, we can produce a {coupled} simulation, $\mathbf y_\hi$, from the density $f_\hi(\mathbf y_\hi \mid \theta, \mathbf y_\lo)$ implied by
\[
f_\mf(\mathbf y_\lo, \mathbf y_\hi \mid \theta) = f_\hi(\mathbf y_\hi \mid \theta, \mathbf y_\lo) f_\lo(\mathbf y_\lo \mid \theta).
\]
Coupling imposes correlations between the resulting likelihood-free weights, $\omega_\lo(\theta, \mathbf y_\lo)$ and $\omega_\hi(\theta, \mathbf y_\hi)$, which thus allows evaluated values of $\omega_\lo$ to provide more information about unknown values of $\omega_\hi$, thereby acting as a variance reduction technique~\citep{Owen2013}.
Given the (coupled) multifidelity model, we can calculate a {multifidelity} likelihood-free weighting as follows.

\begin{defn}
\label{defn:mf}
Let $M$ be any non-negative integer-valued random variable, with conditional probability mass function $p(\cdot \mid \theta, \mathbf y_\lo)$, and with a positive conditional mean,
\[
\mu(\theta, \mathbf y_\lo) = \mathbf E(M \mid \theta, \mathbf y_\lo) > 0.
\]
Given a parameter value, $\theta$, we define
\begin{subequations}
\label{eq:mf}
\begin{align}
\mathbf z &= (\mathbf y_\lo, \mathbf y_{\hi, 1}, \mathbf y_{\hi, 2}, \dots, \mathbf y_{\hi, m}), \label{mf:z} \\
\mathbf y_\lo &\sim f_\lo(\cdot\mid\theta), \label{mf:f_lo} \\
m &\sim p(\cdot \mid \theta, \mathbf y_\lo), \label{mf:p} \\
\mathbf y_{\hi,i} &\sim f_\hi(\cdot\mid\theta,\mathbf y_\lo), \label{mf:f_hi}
\end{align}
noting that each $\mathbf y_{\hi, i}$ may be coupled to the low-fidelity simulation $\mathbf y_\lo \sim f_\lo(\cdot\mid\theta)$.
We combine \Cref{mf:z,mf:f_lo,mf:p,mf:f_hi} to write the density of $\mathbf z$ as  $\phi(\mathbf z \mid \theta)$.
We further define the \emph{multifidelity likelihood-free weighting function}, 
\begin{align}
\omega_\mf(\theta, \mathbf z) &= \omega_\lo(\theta, \mathbf y_\lo) \nonumber \\ 
	&\quad + \frac{1}{\mu(\theta, \mathbf y_\lo)} \sum_{i=1}^m \left[ \omega_\hi(\theta, \mathbf y_{\hi, i}) - \omega_\lo(\theta, \mathbf y_\lo) \right], \label{mf:ell}
\end{align}
as the low-fidelity likelihood-free weighting, corrected by a randomly drawn number, $M=m$, of conditionally independent high-fidelity likelihood-free weightings.
Taking expectations over $\mathbf z$, we write
\begin{equation}
L_\mf(\theta) = \mathbf E(\omega_\mf\mid\theta) = \int \omega_\mf(\theta, \mathbf z) \phi(\mathbf z\mid\theta) ~\mathrm d\mathbf z, \label{mf:L}
\end{equation}
\end{subequations}
as the multifidelity approximation to the likelihood.
\end{defn}

Given $M = m$, only $m$ replicates of $\mathbf y_{\hi,i} \sim f_\hi(\cdot \mid \theta, \mathbf y_\lo)$ need to be simulated for $\omega_\mf(\theta, \mathbf z)$ to be evaluated.
Thus, whenever $m=0$, this means that no high-fidelity simulations need to be completed for $\omega_\mf(\theta, \mathbf z)$ to be calculated, removing the high-fidelity simulation cost from that iteration.
\Cref{alg:mf} presents the adaptation of the basic importance sampling method of \Cref{alg:is} to incorporate the multifidelity weighting function.
The simulation step, $\mathbf y \sim f(\cdot \mid \theta)$, in \Cref{alg:is} is replaced by the \textsc{MF-Simulate} function in \Cref{alg:mf}.

\begin{algorithm}
\caption{Multifidelity likelihood-free importance sampling.}
\label{alg:mf}
\begin{algorithmic}
\Require Prior, $\pi$; importance distribution, $q$; likelihood-free weightings, $\omega_\hi$ and $\omega_\lo$; models $f_\hi(\cdot \mid \theta)$ and $f_\lo(\cdot \mid \theta)$; conditional probability mass function $p(\cdot \mid \theta, \mathbf y_\lo)$ on non-negative integers with mean function $\mu(\theta, \mathbf y_\lo)$; stop condition, \texttt{stop}; target estimated function, $G$.
\Statex
\State Set counter $i=0$.
\Repeat 
	\State Increment counter $i \gets i+1$;
	\State Sample $\theta_i \sim q(\cdot)$;
	\State Generate $\mathbf z_i \sim \phi(\cdot\mid\theta_i)$ from \Call{MF-Simulate}{$\theta_i$};
	\State For $\omega_\mf$ in \Cref{mf:ell}, calculate the weight
		\begin{equation}
		\label{eq:w_mf}
			w_i = w_\mf(\theta, \mathbf z_i) = \frac{\pi(\theta_i)}{q(\theta_i)} \omega_\mf(\theta_i, \mathbf z_i).
		\end{equation}
\Until{\texttt{stop} = \texttt{true}}
\State \Return Weighted sum,
\[
\hat G_\mf = \sum_{i=1}^N w_i G(\theta_i) \bigg/ \sum_{j=1}^N w_j.
\]
\Statex
\Function{MF-Simulate}{$\theta$}
    \State Simulate $\mathbf y_\lo \sim f_\lo(\cdot \mid \theta)$;
	\State Generate $m \sim p(\cdot \mid \theta, \mathbf y_\lo)$ with mean $\mu(\theta, \mathbf y_\lo)$;
	\If{$m=0$}
		\State \Return $\mathbf z = (\mathbf y_\lo)$;
	\Else
		\For{$i=1,\dots,m$}
			\State Simulate $\mathbf y_{\hi, i} \sim f_\hi(\cdot \mid \theta, \mathbf y_\lo)$;
		\EndFor
	    \State \Return $\mathbf z = (\mathbf y_\lo, \mathbf y_{\hi, 1}, \dots, \mathbf y_{\hi, m})$
	\EndIf
\EndFunction
\end{algorithmic}
\end{algorithm}

\subsection{Accuracy of multifidelity inference}
\label{s:hi}

We observe that using $f_\hi$ and $\omega_\hi$ in \Cref{alg:is} produces an estimate of the high-fidelity approximate posterior expectation, $\bar G_{\hi}$.
In \Cref{prop:mf_accuracy}, we show that the multifidelity approximate likelihood, $L_{\mf}(\theta) = \mathbf E(\omega_\mf(\theta, \mathbf z) \mid \theta)$, is equal to the high-fidelity approximate likelihood, $L_{\hi}(\theta) = \mathbf E(\omega_\hi(\theta, \mathbf y_\hi) \mid \theta)$.
As a result, \Cref{alg:mf} also produces a consistent estimate of the high-fidelity approximate posterior expectation, $\bar G_{\hi}$.

\begin{prop}
\label{prop:mf_accuracy}
The multifidelity approximation to the likelihood, $L_{\mf}(\theta) = \mathbf E(\omega_\mf\mid\theta)$, is equal to the high-fidelity approximation to the likelihood, $L_{\hi}(\theta) = \mathbf E(\omega_\hi\mid\theta)$.
Therefore, the estimate $\hat G_\mf$ produced by \Cref{alg:mf} is a consistent estimate of the high-fidelity approximate posterior expectation, $\bar G_{\hi}$.
\end{prop}
\begin{proof}
We take the expectation of $\omega_\mf$ conditional on $(\theta, \mathbf y_\lo, m)$, to find
\[
\mathbf E(\omega_\mf \mid \theta, \mathbf y_\lo, M=m) = \left(1 - \frac{m}{\mu} \right)\omega_\lo(\theta, \mathbf y_\lo) + \frac{m}{\mu} \mathbf E(\omega_\hi \mid \theta, \mathbf y_\lo).
\]
Further taking expectations over the random integer $M$, which has conditional expected value $\mu(\theta, \mathbf y_\lo)$, gives 
\[
\mathbf E(\omega_\mf \mid \theta, \mathbf y_\lo) = \mathbf E(\omega_\hi \mid \theta, \mathbf y_\lo).
\]
Further taking expectations with respect to $\mathbf y_\lo$, it follows that $L_{\omega_\mf}(\theta) = L_{\omega_\hi}(\theta)$.
Therefore, the likelihood-free approximate posteriors, $\pi_{\omega_\mf} = \pi_{\omega_\hi}$, are equal and thus $\hat G_\mf$ is a consistent estimate of $\bar G_{\omega_\mf} = \bar G_{\omega_\hi}$, as required.
\end{proof}

In the limit of infinite computational budgets, the estimate produced by multifidelity importance sampling, in \Cref{alg:mf}, is as accurate as the estimate produced by high-fidelity importance sampling, in \Cref{alg:is} using $f_\hi$ and $\omega_\hi$.
However, we still need to show that the performance of \Cref{alg:mf} exceeds that of \Cref{alg:is} in the practical context of limited computational budgets.
In \Cref{s:performance}, we introduce a method to quantify the performance of \Cref{alg:is,alg:mf} and show that the performance of multifidelity inference is strongly determined by the distribution of $M$, the random number of high-fidelity simulations required at each iteration.

\subsection{Comparing performance}
\label{s:performance}
\Cref{cor:performance} determines the leading-order behaviour of the MSE of the output of \Cref{alg:is} as the computational budget increases.
A similar result applies to the output of \Cref{alg:mf}.
We compare two settings: first, using \Cref{alg:is} with the high-fidelity model, $f_\hi$, and likelihood-free weighting, $\omega_\hi$.
Each iteration has computational cost denoted $C_\hi$, and produces a weighted Monte Carlo sample with weights $w_i$ as independent draws of the random variable $w_\hi$.
The output of \Cref{alg:is} is denoted $\hat G_\hi$, with Monte Carlo error $\hat \Delta_\hi = \hat G_\hi - \bar G_{\hi}$.
The MSE for \Cref{alg:is} has leading-order behaviour
\begin{equation}
\label{eq:perf_hi}
\mathbf E \left( \hat \Delta_\hi^2 \right) 
=
\left[ \frac{\mathbf E(C_\hi)\mathbf E \left(w_\hi^2 \Delta^2\right)}{\mathbf E(w_\hi)^2} \right] \frac{1}{C_{\mathrm{tot}}}
+ O\left( \frac{1}{C_{\mathrm{tot}}^2} \right),
\end{equation}
as the total simulation budget $C_{\mathrm{tot}} \rightarrow \infty$, where $\Delta(\theta) = G(\theta) - \bar G_{\hi}$.

Second, we use \Cref{alg:mf} with the multifidelity model $f_\mf$ and likelihood-free weighting, $\omega_\mf$.
Each iteration has computational cost denoted $C_\mf$, and produces a weighted Monte Carlo sample with weights $w_i$ as independent draws of the random variable $w_\mf$.
The output of \Cref{alg:mf} is $\hat G_\mf$, with Monte Carlo error $\hat \Delta_\mf = \hat G_\mf - \bar G_{\hi}$.
The MSE for \Cref{alg:mf} has leading-order behaviour
\begin{equation}
\label{eq:perf_mf}
\mathbf E \left( \hat \Delta_\mf^2 \right) 
=
\left[ \frac{\mathbf E(C_\mf)\mathbf E \left(w_\mf^2 \Delta^2\right)}{\mathbf E(w_\mf)^2} \right] \frac{1}{C_{\mathrm{tot}}}
+ O\left( \frac{1}{C_{\mathrm{tot}}^2} \right),
\end{equation}
as the total simulation budget $C_{\mathrm{tot}} \rightarrow \infty$, where again $\Delta(\theta) = G(\theta) - \bar G_{\hi}$.

The main result of the paper is given in \Cref{thm:performance}.
For concreteness, we assume that the random variable, $M$, determining the required number of high-fidelity simulations in each iteration of \Cref{alg:mf}, is Poisson distributed, conditional on the parameter value and low-fidelity simulation output.
We show that, for a given multifidelity model and likelihood-free weightings, the mean function, $\mu$, for $M$ determines the performance of \Cref{alg:mf} relative to \Cref{alg:is}.

\begin{thm}
\label{thm:performance}
Assume that the random number of high-fidelity simulations, $M$, required in each iteration of \Cref{alg:mf} is Poisson distributed with conditional mean $\mu(\theta, \mathbf y_\lo)$.
Let $c_\hi(\theta)$
[respectively, $c_\lo(\theta)$ and $c_\hi(\theta, \mathbf y_\lo)$]
be the expected time taken to simulate $\mathbf y_\hi \sim f_\hi(\cdot \mid \theta)$ 
[respectively, to simulate $\mathbf y_\lo \sim f_\lo(\cdot \mid \theta)$ and to produce the coupled high-fidelity simulation $\mathbf y_\hi \sim f_\hi(\cdot \mid \theta, \mathbf y_\lo)$].
Further, assume that the computational cost of each iteration of \Cref{alg:is} and \Cref{alg:mf} can be approximated by the dominant cost of simulation alone, neglecting the costs of the other calculations.

The performance of \Cref{alg:mf}, quantified in \Cref{eq:perf_mf}, exceeds the performance of \Cref{alg:is}, quantified in \Cref{eq:perf_hi}, if and only if $\mathcal J_\mf[\mu] < \mathcal J_\hi$, where
\begin{subequations}
\label{eq:J}
\begin{align}
\mathcal J_\hi &= \bar c_\hi V_\hi, \label{eq:J_hi} \\
\mathcal J_\mf[\mu] &= 
\left( \bar c_\lo + \iint \mu(\theta, \mathbf y_\lo) c_\hi(\theta, \mathbf y_\lo) ~\rho(\theta, \mathbf y_\lo) \mathrm d\theta \mathrm d\mathbf y_\lo \right) \nonumber \\
&\quad{} \times
\left( V_\mf + \iint \Delta_q(\theta)^2 \frac{\eta(\theta, \mathbf y_\lo)}{\mu(\theta, \mathbf y_\lo)} ~\rho(\theta, \mathbf y_\lo) \mathrm d\theta \mathrm d \mathbf y_\lo \right), \label{eq:J_mf}
\end{align}
for the constants
\begin{align}
\bar c_\hi &= \int c_\hi(\theta) ~q(\theta)\mathrm d\theta, \label{eq:c_hi} \\
V_\hi &= \int \Delta_q(\theta)^2 \mathbf E(\omega_\hi^2 \mid \theta) ~q(\theta) \mathrm d\theta, \label{eq:V_hi} \\
\bar c_\lo &= \int c_\lo(\theta) ~q(\theta)\mathrm d\theta, \label{eq:c_lo} \\
V_\mf &= \int \Delta_q(\theta)^2 \mathbf E(\lambda_\hi^2 \mid \theta) ~q(\theta)\mathrm d\theta, \label{eq:V_mf}
\end{align} 
for the functions
\begin{align}
\Delta_q(\theta) &= \frac{\pi(\theta)}{q(\theta)} \left( G(\theta) - \bar G_{\hi} \right), \label{eq:Delta_q} \\
\eta(\theta, \mathbf y_\lo) &= \mathbf E\left(\left(\omega_\hi - \omega_\lo\right)^2 \mid \theta, \mathbf y_\lo \right), \label{eq:eta} \\
\lambda_\hi(\theta, \mathbf y_\lo) &= \mathbf E(\omega_\hi \mid \theta, \mathbf y_\lo), \label{eq:lambda_hi}
\end{align}
and the joint density
\begin{equation}
\rho(\theta, \mathbf y_\lo) = f_\lo(\mathbf y_\lo \mid \theta) q(\theta). \label{eq:rho}
\end{equation}
\end{subequations}
\end{thm}

The performance metrics of \Cref{alg:is}, $\mathcal J_\hi$, and of \Cref{alg:mf}, $\mathcal J_\mf[\mu]$, are each the product of the expected simulation time and the variability of the corresponding likelihood-free weighting.
In the case of \Cref{alg:mf}, the performance depends explicitly on the free choice of the function $\mu(\theta, \mathbf y_\lo)$ that determines the conditional mean of the Poisson-distributed number of high-fidelity simulations required at each iteration.
We observe from the first factor in \Cref{eq:J_mf} that, when $\mu$ is smaller, the total simulation cost is less.
However, the second factor of \Cref{eq:J_mf} implies that as $\mu$ decreases, the variability of the likelihood-free weighting can increase without bound, which can severely damage the performance.
Thus, \Cref{eq:J_mf} illustrates the characteristic multifidelity trade-off between reducing simulation burden while also controlling the increase in sample variance.
Using classical results from calculus of variations, it is possible to determine the mean function, $\mu^\star$, that achieves optimal performance of \Cref{alg:mf}, in the sense of minimising the functional, $\mathcal J_\mf[\mu]$.

\begin{lem}
\label{lem:mustar}
The functional $\mathcal J_\mf[\mu]$ quantifying the performance of \Cref{alg:mf} is optimised by the function $\mu^\star$, where
\begin{equation}
\label{eq:mustar}
\mu^\star(\theta, \mathbf y_\lo)^2 = \Delta_q(\theta)^2 \left[ \frac{\eta(\theta, \mathbf y_\lo) / V_\mf}{c_\hi(\theta, \mathbf y_\lo) / \bar c_\lo} \right].
\end{equation}
\end{lem}

The function $\mu^\star$ given by \Cref{lem:mustar} defines the optimal number of high-fidelity simulations required in any iteration of \Cref{alg:mf}, on average, given the parameter value and low-fidelity simulation output.
We note that larger values of $\eta(\theta, \mathbf y_\lo) = \mathbf E \left( \left( \omega_\hi - \omega_\lo \right)^2 \mid \theta, \mathbf y_\lo \right)$, quantifying the expected squared error between $\omega_\hi$ and $\omega_\lo$, lead to larger values for $\mu^\star$.
Intuitively, if the expected error between the likelihood-free weightings at each fidelity is larger, then the requirement to simulate the high-fidelity model should be greater, to reduce the sample variance.
Conversely, where $c_\hi(\theta, \mathbf y_\lo)$ is larger, the greater simulation time of the high-fidelity model means that $\mu^\star$ should be smaller, and the requirement for the most expensive simulations is reduced.
Intuitively, $\mu^\star$ acts to balance the trade-off between controlling simulation cost and variance identified in \Cref{eq:J_mf} above.

It follows from \Cref{thm:performance} and \Cref{lem:mustar} that \Cref{alg:mf} can only ever be an improvement over \Cref{alg:is} if the optimal performance, $\mathcal J_\mf^\star = \mathcal J_\mf[\mu^\star]$, satisfies $\mathcal J_\mf^\star < \mathcal J_\hi$.

\begin{cor}
\label{cor:muexist}
There exists a mean function, $\mu$, such that the performance of \Cref{alg:mf} exceeds the performance of \Cref{alg:is}, if and only if
\begin{equation}
\sqrt{\frac{\bar c_\lo}{\bar c_\hi} \frac{V_\mf}{V_\hi}} + \iint \sqrt{\frac{\Delta_q(\theta)^2 \eta(\theta, \mathbf y_\lo)}{V_\hi}} \sqrt{\frac{c_\hi(\theta, \mathbf y_\lo)}{\bar c_\hi}} ~\rho(\theta, \mathbf y_\lo) \mathrm d\theta \mathrm d\mathbf y_\lo
< 1.
\label{eq:muexist}
\end{equation}
\end{cor}

The first term in \Cref{eq:muexist} justifies the key assumption that the average computational cost of the low-fidelity model is as small as possible compared to that of the high-fidelity model, $\bar c_\lo < \bar c_\hi$.
The second term is a measure of the total detriment to the performance of \Cref{alg:mf} incurred by the inaccuracy of $\omega_\lo$ versus $\omega_\hi$ as a Monte Carlo estimate of $L_\hi$, as quantified by the function $\eta(\theta, \mathbf y_\lo) = \mathbf E((\omega_\hi - \omega_\lo)^2 \mid \theta, \mathbf y_\lo)$.
This condition justifies two key criteria for the success of the multifidelity method: that low-fidelity simulations are significantly cheaper than high-fidelity simulations, and that the likelihood-free weightings, $\omega_\hi$ and $\omega_\lo$, agree sufficiently well, on average.
A more detailed analysis of \Cref{eq:muexist} is given in \Cref{s:muexist}.

The proofs of \Cref{thm:performance}, \Cref{lem:mustar} and \Cref{cor:muexist} are given in \Cref{s:unconstrained_optimisation}.
However, we note that these analytical results are useful only insofar as the various functions and constants given in \Cref{eq:J} are known. 
In particular, evaluating the optimal mean function, $\mu^\star(\theta, \mathbf y_\lo)$, in \Cref{alg:mf} requires knowledge of the functions $\Delta_q(\theta)$, $\eta(\theta, \mathbf y_\lo)$, and $c_\hi(\theta, \mathbf y_\lo)$, and of the constants $V_\mf$ and $\bar c_\lo$.
Similarly, certifying whether the multifidelity approach can outperform the high-fidelity importance sampling method relies on knowing and integrating these functions.
However, it is typically the case that these functions and constants are unknown a priori, and need to be estimated based on simulations.
In the following section, we describe how the analytical results of \Cref{s:performance} can be used to construct a heuristic adaptive multifidelity algorithm that learns a near-optimal mean function, $\mu(\theta, \mathbf y_\lo)$, as simulations at each fidelity are completed.

\section{Multifidelity implementation}
\label{s:mfimplementation}
In \Cref{s:performance}, we derived the optimal mean function for the Poisson distribution of the number, $M$, of high-fidelity simulations required in an iteration of \Cref{alg:mf}, conditioned on the parameter value, $\theta$, and low-fidelity simulation, $\mathbf y_\lo$.
The optimality condition was based on minimising the functional $\mathcal J_\mf[\mu]$ defined in \Cref{eq:J_mf}, with minimiser $\mu^\star$ given in \Cref{eq:mustar}.
While we can derive the analytical form of $\mu^\star$, this cannot generally be determined a priori, but must be learned in parallel with carrying out likelihood-free inference.

In this section, we describe a practical approach to determining a near-optimal mean function for use in multifidelity likelihood-free inference.
We rely on two approximations, relative to the analytically optimal mean function $\mu^\star$ given in \Cref{eq:mustar}.
First, we constrain the optimisation of $\mathcal J_\mf$ to the space of functions, $\mu_{\mathcal D}$, that are piecewise constant in an arbitrary, given, finite partition, $\mathcal D$, of the global space of $(\theta,\mathbf y_\lo)$ values.
The resulting optimisation problem is therefore finite-dimensional.
However, although this optimisation can be solved analytically, we can observe that its estimation, being based on the ratios of simulation-based Monte Carlo estimates, is numerically unstable.
This motivates a second approximation, which is to follow a gradient-descent approach to allow the mean function to adaptively converge towards the optimum.

\subsection{Piecewise constant assumption}
We constrain the space of mean functions, $\mu$, to be piecewise constant.
Consider an arbitrary, given collection $\mathcal D = \{ D_k \mid k=1,\dots,K \}$ of $\rho$-integrable sets that partition the global space of $(\theta, \mathbf y_\lo)$ values.
We denote a $\mathcal D$-piecewise constant function, parametrised by the vector $\nu = (\nu_1, \dots, \nu_k)$, as
\[
\mu_{\mathcal D}(\theta, \mathbf y_\lo;~\nu) = \sum_{k=1}^K \nu_k \mathbf I((\theta, \mathbf y_\lo) \in D_k).
\]
Substituting this function into \Cref{eq:J_mf}, we can quantify the performance of \Cref{alg:mf}, using the mean defined by $\mu_{\mathcal D}(\theta, \mathbf y_\lo;~\nu)$, as the parametrised product,
\begin{subequations}
\label{eq:J_mf_D}
\begin{equation}
\mathcal J_{\mathcal D}(\nu) = \left( \bar c_\lo + \sum_{k=1}^K c_k \nu_k \right) \left( V_\mf + \sum_{k=1}^K  \frac{V_k}{\nu_k} \right),
\end{equation}
with coefficients given by the integrals
\begin{align}
c_k &= \iint_{D_k} c_\hi(\theta, \mathbf y_\lo) \rho(\theta, \mathbf y_\lo) ~\mathrm d \theta \mathrm d \mathbf y_\lo, \\
V_k &= \iint_{D_k} \Delta_q(\theta)^2 \eta(\theta, \mathbf y_\lo) \rho(\theta, \mathbf y_\lo) ~\mathrm d \theta \mathrm d \mathbf y_\lo.
\end{align}
\end{subequations}
Similarly to the functional $\mathcal J_\mf[\mu]$ optimised in \Cref{lem:mustar} across positive functions, $\mu$, we can optimise the function $\mathcal J_{\mathcal D}(\nu)$ across positive vectors, $\nu$.
\begin{lem}
\label{lem:nustar}
The optimal function values $\nu_k$ for $\mu_{\mathcal D}(\theta, \mathbf y_\lo;~\nu)$ that minimise $\mathcal J_{\mathcal D}(\nu)$ are
\begin{equation}
\label{eq:nustar}
\nu_k^\star = \sqrt{\frac{V_k / V_\mf}{c_k / \bar c_\lo}},
\end{equation}
for $c_k$ and $V_k$ as defined in \Cref{eq:J_mf_D} and for $V_\mf$ and $\bar c_\lo$ as given in \Cref{eq:V_mf,eq:c_lo}.
The resulting performance of \Cref{alg:mf} with this mean function is
\[
\mathcal J_{\mathcal D}^\star = \mathcal J_{\mathcal D}(\nu^\star) = \left( \sqrt{\bar c_\lo V_\mf} + \sum_{k=1}^K \sqrt{c_k V_k}  \right)^2.
\]
\end{lem}
Similarly to the analytical results of \Cref{s:performance}, to evaluate $\nu^\star$ we need to estimate the values of $\bar c_\lo$, $c_k$, $V_\mf$ and $V_k$, which are unknown a priori.
Although these values can be estimated based on Monte Carlo simulation, the rational form of $\nu_k^\star$ means that these estimates can be unstable, particularly for sets $D_k \in \mathcal D$ with small volume (measured by the density, $\rho$).
We now consider a conservative approach to determining values for $\nu$ that will provide stable estimates of $\nu^\star$.

\subsection{Adaptive multifidelity likelihood-free inference}
\label{s:adaptive}
Rather than directly targeting $\nu^\star$, based on ratios of highly variable Monte Carlo estimates, we can introduce a gradient-descent approach to updating the vector $\nu$.
Taking derivatives of $\mathcal J_{\mathcal D}$ with respect to $\log \nu_k$ for $k=1,\dots,K$ gives the gradient,
\begin{equation*}
\frac{\partial \mathcal J_{\mathcal D}}{\partial \log \nu_k} = \nu_k c_k \left( V_\mf + \sum_{j=1}^K  \frac{V_j}{\nu_j} \right) - \frac{V_k}{\nu_k} \left( \bar c_\lo + \sum_{j=1}^K c_j \nu_j \right).
\end{equation*}
Thus, if we write $\nu^{(r)}$ for the value of $\nu$ used in iteration $r$ of \Cref{alg:mf}, we intend to update to $\nu^{(r+1)}$ in the next iteration using gradient descent, such that
\begin{equation}
\label{eq:nu}
\log \nu^{(r+1)}_k = \log \nu^{(r)}_k - \delta \left[
	\nu^{(r)}_k c_k \left( V_\mf + \sum_{j=1}^K  \frac{V_j}{\nu^{(r)}_j} \right) 
	- \frac{V_k}{\nu^{(r)}_k} \left( \bar c_\lo + \sum_{j=1}^K  c_j \nu^{(r)}_j \right)
\right].
\end{equation}
Note that we express this updating rule in terms of $\log \nu_k^{(r)}$ to ensure that each $\nu_k^{(r)}$ is positive, since the updates to $\nu^{(r)}$ are multiplicative.
As is typical of gradient-descent approaches, \Cref{eq:nu} requires the specification of the step size hyperparameter, $\delta$.
It is straightforward to show that $\nu^\star$ is the unique positive stationary point of \Cref{eq:nu}.
Furthermore, since each derivative $\partial \mathcal J_{\mathcal D} / \partial \log \nu_k$ is quadratic in the variables $c_j$, $V_j$, $\bar c_\lo$ and $V_\mf$, the numerical instability in estimating \Cref{eq:nustar} as a ratio does not occur when estimating these derivatives.
In relatively undersampled regions $D_k \in \mathcal D$ with small $\rho$-volume, the small values of $c_k$ and $V_k$ ensure that the convergence to the corresponding estimated optimal value, $\nu_k^\star$, is more conservative.

We now explicitly set out the Monte Carlo estimates of $c_j$, $V_j$, $\bar c_\lo$ and $V_\mf$.
These estimates can then be substituted into \Cref{eq:nu} to produce an updating rule for $\nu_k^{(r)}$.
This adaptive approach is then implemented into multifidelity likelihood-free importance sampling, as described in \Cref{alg:mf2}.

\begin{lem}
\label{lem:numc}
Suppose that $r$ iterations of \Cref{alg:mf} have been completed.
We denote:
\begin{itemize}
\item $c_{\lo, i}$ for the observed simulation cost of each $\mathbf y_{\lo, i}$;
\item $\omega_{\lo, i} = \omega_\lo(\theta_i, \mathbf y_{\lo, i})$ for the low-fidelity weighting calculated at each iteration;
\item $\mu_i = \mu^{(i)}(\theta_i, \mathbf y_{\lo,i})$ for the value of the mean function used to specify the random variable $M$ in iteration $i$; 
\item $m_i$ for the randomly drawn value of $M$ in iteration $i$, with mean $\mu_i$;
\item $c_{\hi, i, j}$ for the observed simulation cost of each $\mathbf y_{\hi, i, j}$ for $j = 1, \dots, m_i$ as the values of the $m_i$ high-fidelity weightings calculated in iteration $i$, noting that $m_i$ may be zero;
\item $\omega_{\hi, i, j} = \omega_\hi(\theta_i, \mathbf y_{\hi, i, j})$ for $j = 1, \dots, m_i$ as the values of the $m_i$ high-fidelity weightings calculated in iteration $i$; 
\item $\bar G_\hi^{(r)}$ as the current Monte Carlo estimate of $\bar G_\hi$;
\item $\Delta^{(r)}_i = \left[ G(\theta_i) - \bar G_\hi^{(r)} \right] \pi(\theta_i) / q(\theta_i)$ as the importance weighting, centred on the current estimate of $\bar G_\hi$.
\end{itemize}
The simulation-based Monte Carlo quantities
\begin{subequations}
\label{eq:numc}
\begin{align}
\bar c_\lo^{(r)} &= \frac{1}{r} \sum_{i=1}^r c_{\lo, i},
\\
V_{\mf}^{(r)} &= \frac{1}{r} \sum_{i=1}^r \left( \frac{\Delta_i^{(r)}}{\mu_i} \right)^2 \left[ \left( \sum_{j=1}^{m_i} \omega_{\hi, i, j} \right)^2 - \sum_{j=1}^{m_i} \omega_{\hi, i, j}^2 \right],
\\
c_k^{(r)} &= \frac{1}{r} \sum_{i=1}^r \mathbf I_{D_k}(\theta_i, \mathbf y_{\lo, i}) \frac{1}{\mu_i} \sum_{j=1}^{m_i} c_{\hi, i, j},
\\
V_k^{(r)} &= \frac{1}{r} \sum_{i=1}^r \mathbf I_{D_k}(\theta_i, \mathbf y_{\lo, i}) \frac{1}{\mu_i} \sum_{j=1}^{m_i} \left( \Delta_i^{(r)} \left( \omega_{\hi, i, j} - \omega_{\lo, i} \right) \right)^2,
\end{align}
\end{subequations}
are consistent estimates of $\bar c_\lo$, $V_\mf$, $c_k$, and $V_k$, respectively.
\end{lem}

Substituting the estimates in \Cref{eq:numc} into the updating rule, \Cref{eq:nu}, we use 
\begin{equation}
\label{eq:update}
\log \nu_k^{(r+1)} = \log \nu_k^{(r)} - \delta \left[
	\nu^{(r)}_k c_k^{(r)} \left( V_\mf^{(r)} + \sum_{j=1}^K  \frac{V_j^{(r)}}{\nu^{(r)}_j} \right) 
	- \frac{V_k^{(r)}}{\nu^{(r)}_k} \left( \bar c_\lo^{(r)} + \sum_{j=1}^K  c_j^{(r)} \nu^{(r)}_j \right),
\right]
\end{equation}
to update $\nu^{(r)}$ to $\nu^{(r+1)}$ in adaptive multifidelity likelihood-free importance sampling, as outlined in \Cref{alg:mf2}.
In addition to the specification of the step size hyperparameter, $\delta$, \Cref{alg:mf2} also requires a burn-in phase, $N_0$, to initialise the Monte Carlo estimates in \Cref{eq:numc}.
The partition, $\mathcal D = \left\{ D_1, \dots, D_K \right\}$, is also an input into \Cref{alg:mf2}.
We defer an investigation of how to choose this partition to future work.
For the purposes of this paper, however, we can heuristically construct a partition, $\mathcal D$, by fitting a decision tree.
We use the burn-in phase of \Cref{alg:mf2}, over iterations $i \leq N_0$, and regress the values of
\[
\mu^\star_i = \left|\Delta_i^{(N_0)}\right| \sqrt{\frac{\sum_j (\omega_{\hi,i,j} - \omega_{\lo,i})^2}{\sum_j c_{\hi, i, j}}},
\]
against features $(\theta_i, \mathbf y_{\lo,i})$, using the CART algorithm \citep{Hastie2009} as implemented in \texttt{DecisionTrees.jl}.
Note that this regression is motivated by the form of the true optimal mean function, $\mu^\star$, given in \Cref{eq:mustar}.
The resulting decision tree defines a partition, $\mathcal D = \{ D_1, \dots, D_K \}$, used to define the piecewise-constant mean function $\mu_{\mathcal D}(\theta, y_\lo;~\nu)$ over $i > N_0$.

\begin{algorithm}
\caption{Adaptive multifidelity likelihood-free importance sampling.}
\label{alg:mf2}
\begin{algorithmic}
\Require Prior, $\pi$; importance distribution, $q$; likelihood-free weightings, $\omega_\hi$ and $\omega_\lo$; models $f_\hi(\cdot \mid \theta)$ and $f_\lo(\cdot \mid \theta)$; partition $\mathcal D = \{ D_1, \dots, D_K \}$ of $(\theta, \mathbf y_\lo)$ space; adaptation rate, $\delta$; burn-in period, $N_0$; stop condition, \texttt{stop}; target estimated function, $G$.
\Statex
\State Set counter $i=0$.
\State Initialise $\log \nu_k^{(1)} = 0$ for $k=1,\dots,K$.
\Repeat 
	\State Increment counter $i \gets i+1$;
	\State Sample $\theta_i \sim q(\cdot)$;
	\State Generate $\mathbf z_i \sim \phi(\cdot\mid\theta_i)$ from \Call{MF-Simulate}{$\theta_i, \nu^{(i)}$};
	\State For $\omega_\mf$ in \Cref{mf:ell}, calculate the weight
		\begin{equation}
		\label{eq:w_mf2}
			w_i = w_\mf(\theta, \mathbf z_i) = \frac{\pi(\theta_i)}{q(\theta_i)} \omega_\mf(\theta_i, \mathbf z_i).
		\end{equation}
	\If{$i>N_0$}
		\State Generate $\nu^{(i+1)}$ from \Call{Update-Nu}{$\nu^{(i)}$}
	\Else
		\State Set $\nu^{(i+1)} = \nu^{(i)}$.
	\EndIf
\Until{\texttt{stop} = \texttt{true}}
\State \Return Weighted sum,
\[
\hat G_\mf = \sum_{i=1}^N w_i G(\theta_i) \bigg/ \sum_{j=1}^N w_j.
\]
\Statex
\Function{MF-Simulate}{$\theta, \nu$}
    \State Simulate $\mathbf y_\lo \sim f_\lo(\cdot \mid \theta)$;
	\State Find $k$ such that $(\theta, \mathbf y_\lo) \in D_k$;
	\State Generate $m \sim \mathrm{Poi}(\nu_k)$;
	\For{$i=1,\dots,m$}
		\State Simulate $\mathbf y_{\hi, i} \sim f_\hi(\cdot \mid \theta, \mathbf y_\lo)$;
	\EndFor
    \State \Return $\mathbf z = (\mathbf y_\lo, m, \mathbf y_{\hi, 1}, \dots, \mathbf y_{\hi, m})$
\EndFunction
\Statex
\Function{Update-Nu}{$\nu_1, \dots, \nu_K$}
	\State Update Monte Carlo estimates defined in \Cref{eq:numc};
	\For{$k=1,\dots,K$}
		\State Increment $\log \nu_k$ according to \Cref{eq:update};
	\EndFor
	\State \Return $\nu = (\nu_1, \dots, \nu_k)$
\EndFunction
\end{algorithmic}
\end{algorithm}

\section{Example: Biochemical reaction network}
\label{s:eg}
The following example considers the stochastic simulation of a biochemical reaction motif.
Readers unfamiliar with these techniques are referred to detailed expositions by \citet{Warne2019} and \citet{Erban2019}.
We model the conversion (over time $t \geq 0$) of substrate molecules, labelled $\mathrm S$, into molecules of a product, $\mathrm P$.
The conversion of $\mathrm S$ into $\mathrm P$ is catalysed by the presence of enzyme molecules, $\mathrm E$, which bind with $\mathrm S$ to form molecules of complex, labelled $\mathrm C$.
After non-dimensionalising units of time and volume, this network motif is represented by three reactions,
\begin{subequations}
\label{eq:enzyme_hi}
\begin{equation} 
\label{eq:enzyme_hi_net}
\ce{S + E <=>[$k_1$][$k_2$] C ->[$k_3$] P + E}, 
\end{equation}
parametrised by the vector $\theta = (k_1, k_{-1}, k_2)$ of positive parameters, $k_1$, $k_{-1}$, and $k_2$, which define three propensity functions,
\begin{align}
v_1(t) &= k_1 S(t) E(t), \\
v_2(t) &= k_{-1} C(t), \\
v_3(t) &= k_2 C(t),
\end{align}
\end{subequations}
where the integer-valued variables $S(t)$, $E(t)$, $C(t)$ and $P(t)$ represent the molecule numbers at time $t>0$.
At $t=0$, we assume there are no complex or product molecules, but set positive integer numbers $S_0 = 100$ and $E_0 = 5$ of substrate and enzyme molecules, respectively.
Given the fixed initial conditions, the parameters in $\theta$ are sufficient to specify the dynamics of the model in \Cref{eq:enzyme_hi_net}.
The model is stochastic, and induces a distribution, which we denote $f(\cdot \mid \theta)$, on the space of trajectories $x:t \mapsto (S(t), E(t), C(t), P(t))$ of molecule numbers in $\mathbb N^4$ over time.

For the purposes of this example, the observed data 
\[
y_0 = (y_1, \dots, y_{10}) = (1.73, 3.80, 5.95, 8.10, 11.17, 12.92, 15.50, 17.75, 20.17, 23.67),
\]
depicted in \Cref{fig:trajfig}, are the times at which the number of product molecules reaches $P(y_n) = 10n$.
We set a prior $\pi(\theta)$ on the vector $\theta$, equal to a product of independent uniform distributions such that $k_1, k_{-1} \sim \mathrm U(10, 100)$ and $k_2 \sim \mathrm U(0.1, 10)$.
We seek the posterior distribution $\pi(\theta \mid y_0)$ using the likelihood, denoted $\mathcal L(\theta) = f(y_0 \mid \theta)$, focusing on the posterior expectation of the function $G(\theta) = k_2$, denoting the rate of conversion of substrate--enzyme complex to product.

All code for this example is available at \href{https://github.com/tpprescott/mf-lf}{\texttt{github.com/tpprescott/mf-lf}}, using stochastic simulations implemented by \href{https://github.com/tpprescott/ReactionNetworks.jl}{\texttt{github.com/tpprescott/ReactionNetworks.jl}}.

\subsection{Multifidelity approximate Bayesian computation}
\label{s:mfabc}
\subsubsection{ABC importance sampling}
We assume that we cannot calculate the likelihood function, $\mathcal L(\theta) = f(y_0 \mid \theta)$.
Instead, we need to use simulations to perform ABC.
Given $\theta$, the model in \Cref{eq:enzyme_hi} can be exactly simulated using the Gillespie stochastic simulation algorithm, to produce draws $y \sim f(\cdot \mid \theta)$ from the exact model \citep{Gillespie1977,Erban2019,Warne2020}.
We will use the ABC likelihood-free weighting with threshold value $\epsilon = 5$ on the Euclidean distance of the simulation from $y_0$, such that
\[
\omega(\theta, y) = \mathbf 1(\|y - y_0\|_2 < 5),
\]
to define the likelihood-free approximation to the posterior, $L_\ABC(\theta) = \mathbf E(\omega \mid \theta)$.
We combine this likelihood-free weighting in \Cref{alg:is} with a rejection sampling approach, setting the importance distribution $q = \pi$ equal to the prior.

\subsubsection{Multifidelity ABC}
The exact Gillespie stochastic simulation algorithm can incur significant computational burden. 
In the specific case of the network in \Cref{eq:enzyme_hi}, if the reaction rates $k_{\pm 1}$ are large relative to $k_2$, there are large numbers of binding/unbinding reactions $\ce{S + E <-> C}$ that occur in any simulation.
In comparison, the reaction $\ce{C -> P + E}$ can only fire exactly $100$ times.
Michaelis--Menten dynamics exploit this scale separation to approximate the enzyme kinetics network motif.
We approximate the conversion of substrate into product as a single reaction step,
\begin{subequations}
\label{eq:enzyme_lo}
\begin{equation}
\label{eq:enzyme_lo_net}
\ce{S ->[$k_{\mathrm{MM}}(t)$] P},
\end{equation}
where the time-varying rate of conversion, $k_{\mathrm{MM}}(t)$, given by
\begin{align}
k_{\mathrm{MM}}(t) &= \frac{k_2 \min(S(t), E_0)}{K_{\mathrm{MM}} + S(t)}, \\
K_{\mathrm{MM}} &= \left( k_{-1} + k_2 \right) / k_1,
\end{align}
\end{subequations}
induces the propensity function $v_{\mathrm{MM}}(t) = k_{\mathrm{MM}}(t) S(t)$.
We assume initial conditions of $S(0)=S_0=100$ and $P(0)=0$, and fix the parameter $E_0 = 5$.
Thus, the parameter vector, $\theta = (k_1, k_{-1}, k_2)$, again fully determines the dynamics of the low-fidelity model in \Cref{eq:enzyme_lo}.
We write $f_\lo(y_\lo \mid \theta)$ as the conditional probability density for the Gillespie simulation of the approximated model in \Cref{eq:enzyme_lo}, where $y_\lo$ is the vector of ten simulated time points $y_{\lo, n}$ at which $10n$ product molecules have been produced.

For a biochemical reaction network consisting of $R$ reactions, the Gillespie simulation algorithm is a deterministic transformation of $R$ independent unit-rate Poisson processes, one for each reaction channel.
We can couple the models in \Cref{eq:enzyme_hi,eq:enzyme_lo} by using the same Poisson process for the single reaction in \Cref{eq:enzyme_lo} and for the product formation \ce{C -> P + E} reaction of \Cref{eq:enzyme_hi} \citep{Prescott2020,Lester2019}.
Using this coupling approach, we first simulate $y_\lo \sim f_\lo(\cdot \mid \theta)$ from \Cref{eq:enzyme_lo}.
We then produce the coupled simulation $y_\hi \sim f_\hi(\cdot \mid \theta, y_\lo)$ from the model in \Cref{eq:enzyme_hi}, using the shared Poisson process.
We set the corresponding likelihood-free weightings to
\begin{align*}
\omega_\hi(\theta, y_\hi) = \mathbf I(|y_\hi-y_0|<5), \\
\omega_\lo(\theta, y_\lo) = \mathbf I(|y_\lo-y_0|<5),
\end{align*}
noting that $\mathbf E(\omega_\hi \mid \theta) = L_\ABC(\theta)$ is the high-fidelity ABC approximation to the likelihood.
\Cref{fig:trajfig} illustrates the effect of coupling between low-fidelity and high-fidelity models.
The five coupled high-fidelity simulations are significantly less variable than the independent high-fidelity simulations, appearing almost coincident in \Cref{fig:trajfig}.
This ensures a large degree of correlation between the coupled likelihood-free weightings, $\omega_\hi$ and $\omega_\lo$.
Thus, coupling ensures that $\omega_\lo$ is a reliable proxy for $\omega_\hi$ for use in multifidelity likelihood-free inference.

\begin{figure}
\centering
\includegraphics[width=\textwidth]{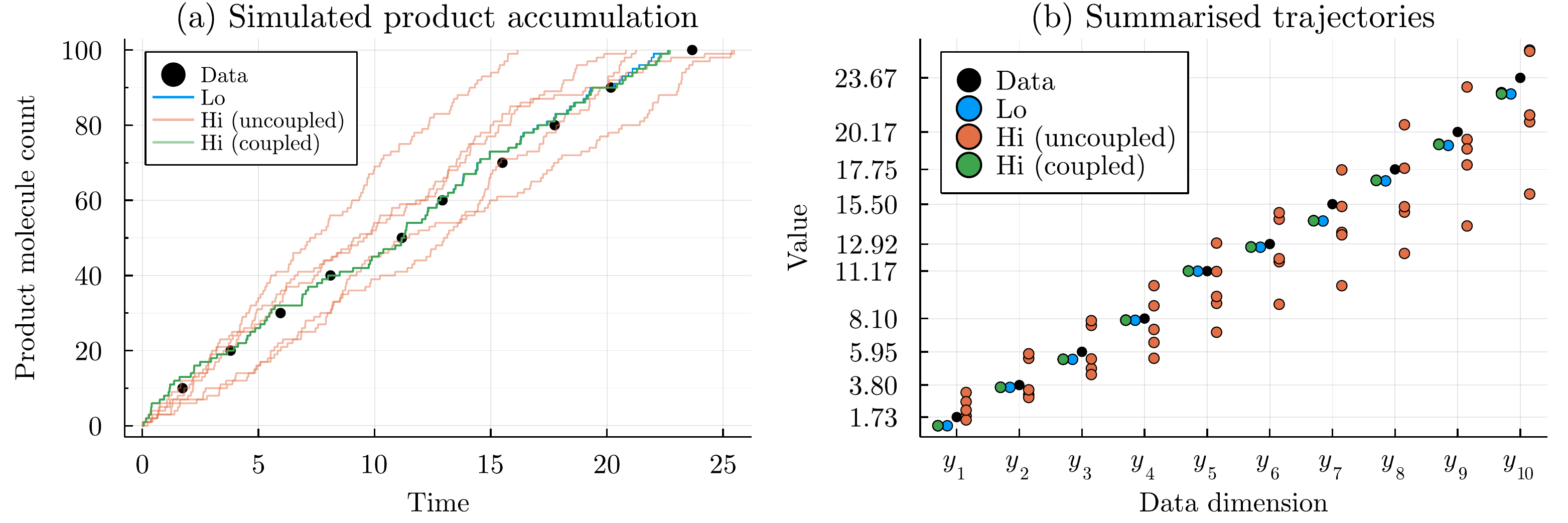}
\caption{%
	Effect of multifidelity coupling.
	(a) Example stochastic trajectories from the high and low-fidelity enzyme kinetics models in \Cref{eq:enzyme_hi,eq:enzyme_lo} for parameters $\theta = (k_1, k_{-1}, k_2) = (50, 50, 1)$, compared with data used for inference.
	For one low-fidelity simulation, we generate five uncoupled simulations and five coupled simulations.
	(b) Ten-dimensional data summarising simulated trajectories in (a). Black represents observed data, $y_0$; the single low-fidelity simulation $y_\lo \sim f_\lo(\cdot \mid \theta)$ is in blue; five uncoupled simulations $y_\hi \sim f_\hi(\cdot \mid \theta)$ are in orange; five coupled simulations $y_\hi \sim f_\hi(\cdot \mid \theta, y_\lo)$ are in green (almost coincident).
}%
\label{fig:trajfig}
\end{figure}

We implement \Cref{alg:mf2} by setting a burn-in period of $N_0 = 10,000$, for which we generate $m_i \sim M = \mathrm{Poi}(1)$ high-fidelity simulations at each iteration, $i \leq N_0$.
Once the burn-in period is complete, we define the partition $\mathcal D$ by learning a decision tree through a simple regression, as described in \Cref{s:mfimplementation}.
For iterations $i>N_0$ beyond the burn-in period, we set a step size of $\delta=10^3$ for the gradient descent update in \Cref{eq:nu}.

\subsubsection{Results}
\Cref{alg:is} was run four times, setting the \texttt{stop} condition to $i=10,000$, $i=20,000$, $i=40,000$ and $i=80,000$.
Similarly, \Cref{alg:mf2} was run five times, setting the \texttt{stop} condition to $i=40,000$, $i=80,000$, $i=160,000$, $i=320,000$ and $i=640,000$.
\Cref{fig:ABC}a shows how the variance in the estimate, $\hat G$, varies with the total simulation cost, $C_{\mathrm{tot}}$, shown for each of the two algorithms.
The slope of each curve (on a log-log scale) is approximately $-1$, corresponding to the dominant behaviour of the MSE being reciprocal with total simulation time, as observed in \Cref{eq:perf_mf}.
The offset in the two curves corresponds to the inequality $\mathcal J_\mf < \mathcal J_\hi$ in the leading order coefficient, thereby demonstrating the improved performance of \Cref{alg:mf2} over \Cref{alg:is}.

\begin{figure}
\centering
\includegraphics[width=\textwidth]{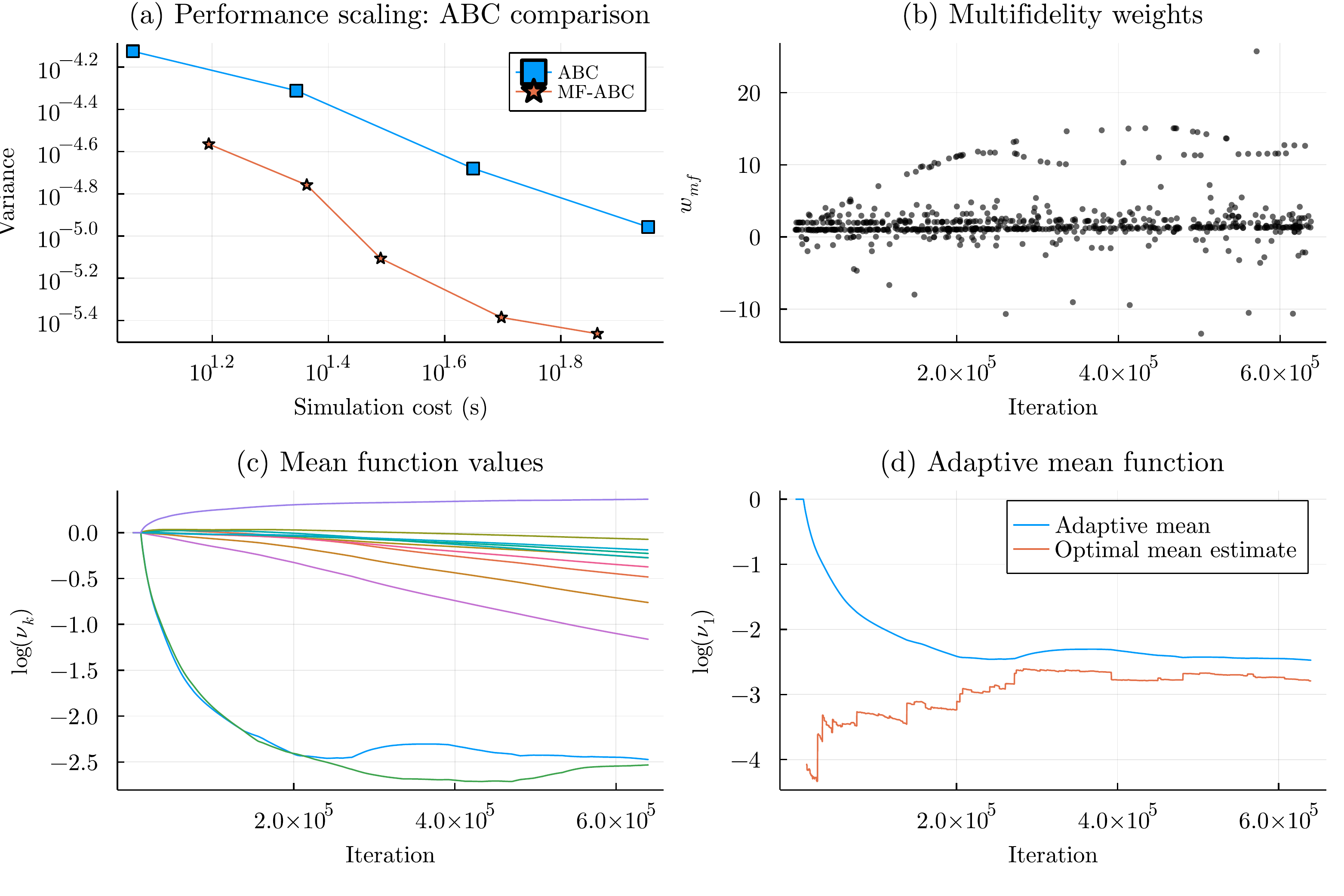}
\caption{%
	Multifidelity ABC.
	(a) Total simulation cost versus estimated variance of output estimate, $\hat G$, for four runs of \Cref{alg:is} (ABC) and five runs of \Cref{alg:mf2} (MF-ABC).
	(b) Values of the multifidelity weight, $w_\mf$, during the longest run of \Cref{alg:mf2}, for iterations where $\omega_\mf \neq \omega_\lo$, such that the low-fidelity likelihood-free weighting is corrected based on at least one high-fidelity simulation.
	(c) The evolution of the values of $\nu_k^{(i)}$ during the longest run of \Cref{alg:mf2}.
	(d) A comparison of the adaptive $\nu_1^{(i)}$ to the evolving best estimate of the optimal $\nu_1^\star$, given by \Cref{eq:nustar}, based on the Monte Carlo estimates in \Cref{eq:numc}.
}%
\label{fig:ABC}
\end{figure}

The values in \Cref{fig:ABC}b show the multifidelity weights, $w_i$.
We show only those weights not equal to zero or one, corresponding to those iterations where $\omega_\lo(\theta_i, y_{\lo,i})$ has been corrected by at least one $\omega_\hi(\theta_i, y_{\hi,i,j}) \neq \omega_\lo(\theta_i, y_{\lo,i})$.
Clearly there is a significant amount of correction applied to the low-fidelity weights.
However, as demonstrated by the improved performance statistics, \Cref{alg:mf2} has learned the required allocation of computational budget to the high-fidelity simulations that balances the trade-off between achieving reduced overall simulation times and correcting inaccuracies in the low-fidelity simulation.

Each run of \Cref{alg:mf2} includes a burn-in period of $10,000$ iterations, at the conclusion of which a partition $\mathcal D$ is created, based on decision tree regression.
In \Cref{s:mean_eg}, we show how this decision tree is used to define a piecewise-constant mean function, specifically for the partition $\mathcal D$ used for the final run of \Cref{alg:mf2} (i.e. for stopping condition $i=640,000$).
In \Cref{fig:ABC}c, we show the evolution of the values of $\nu_k^{(i)}$ used in this mean function, over iterations $i$.
Following the updating rule in \Cref{eq:update}, the trajectory of $\nu_k^{(i)}$ converges exponentially towards a Monte Carlo estimate of the optimal value $\nu_k^\star$ given in \Cref{eq:nustar}.
However, we can see from \Cref{fig:ABC}c that, as more simulations are completed and the Monte Carlo estimates in \Cref{eq:numc} evolve, the values of each parameter, $\nu_k$, track updated estimates.
This is illustrated in \Cref{fig:ABC}d for $\nu_1$, where the estimated optimum $\nu_1^\star$ evolves as more simulations are completed.
We note that the gradient descent update in \Cref{eq:update} at iteration $i$ depends on \emph{all} $\nu_k^{(i)}$ values.
Thus, the observed convergence of $\nu_1^{(i)}$ to the evolving estimate of $\nu_1^\star$ is not necessarily monotonic.

\Cref{fig:ABC}d illustrates the motivation for the use of gradient descent rather than simply using the analytically obtained optimum.
When very few simulations have been completed, then the estimates in \Cref{eq:numc} are small and their ratios are numerically unstable, and often far from the true optimum.
If $\nu_k^{(i)}$ values are too small in early iterations, then estimates become \emph{more} numerically unstable, since fewer high-fidelity simulations are completed for small $\mu$.
Instead, using gradient descent ensures that enough high-fidelity simulations are completed for each $\mathcal D_k$, including those with low volume under the measure $\rho$, to stabilise the estimates in \Cref{eq:numc} and thus stabilise the multifidelity algorithm.



\subsection{Multifidelity Bayesian synthetic likelihood}\label{s:bsl}
Consider the same model of enzyme kinetics as in \Cref{s:mfabc}.
As depicted in \Cref{fig:trajfig}, this model has low-fidelity (Michaelis--Menten) stochastic dynamics with distribution $f_\lo(\cdot\mid\theta)$, and coupled high-fidelity stochastic dynamics with distribution $f_\hi(\cdot\mid\theta,y_\lo)$.
We now redefine $\omega_\lo$ and $\omega_\hi$ to be Bayesian synthetic likelihoods, based on $K$ pairs of coupled simulations, 
\begin{align*}
y_{\lo,k} &\sim f_\lo(\cdot \mid \theta), \\
y_{\hi,k} &\sim f_\hi(\cdot \mid \theta, y_{\lo,k}),
\end{align*}
for $k=1,\dots,K$. That is,
\begin{align*}
\omega_\lo(\theta, \mathbf y_\lo) &= \mathcal N \left( y_0 : \mu(\mathbf y_\lo), \Sigma(\mathbf y_\lo) \right), \\
\omega_\hi(\theta, \mathbf y_\hi) &= \mathcal N \left( y_0 : \mu(\mathbf y_\hi), \Sigma(\mathbf y_\hi) \right),
\end{align*}
are the Gaussian likelihoods of the observed data, under the empirical mean and covariance of $K$ low-fidelity and (coupled) high-fidelity simulations, respectively.

\Cref{alg:is} was run three times, using $\omega_\hi(\theta, \mathbf y_\hi)$ dependent on high-fidelity simulations $\mathbf y_\hi \sim f(\cdot \mid \theta)$, alone, and setting the \texttt{stop} condition to $i=2,500$, $i=5,000$ and $i=10,000$.
Similarly, \Cref{alg:mf2} was run four times using the coupled multifidelity model, setting the \texttt{stop} condition to $i=4,000$, $i=8,000$, $i=16,000$ and $i=32,000$, and initialising with a burn-in of size $N_0 = 2,000$.
The adaptive step size is set to $\delta = 10^8$.
In both algorithms, we set the number of simulations required for each evaluation of $\omega_\hi(\theta, (y_{\hi,1},\dots,y_{\hi,K}))$ or $\omega_\lo(\theta, (y_{\lo,1},\dots,y_{\lo,K}))$ as $K=100$.

\begin{figure}
\centering
\includegraphics[width=\textwidth]{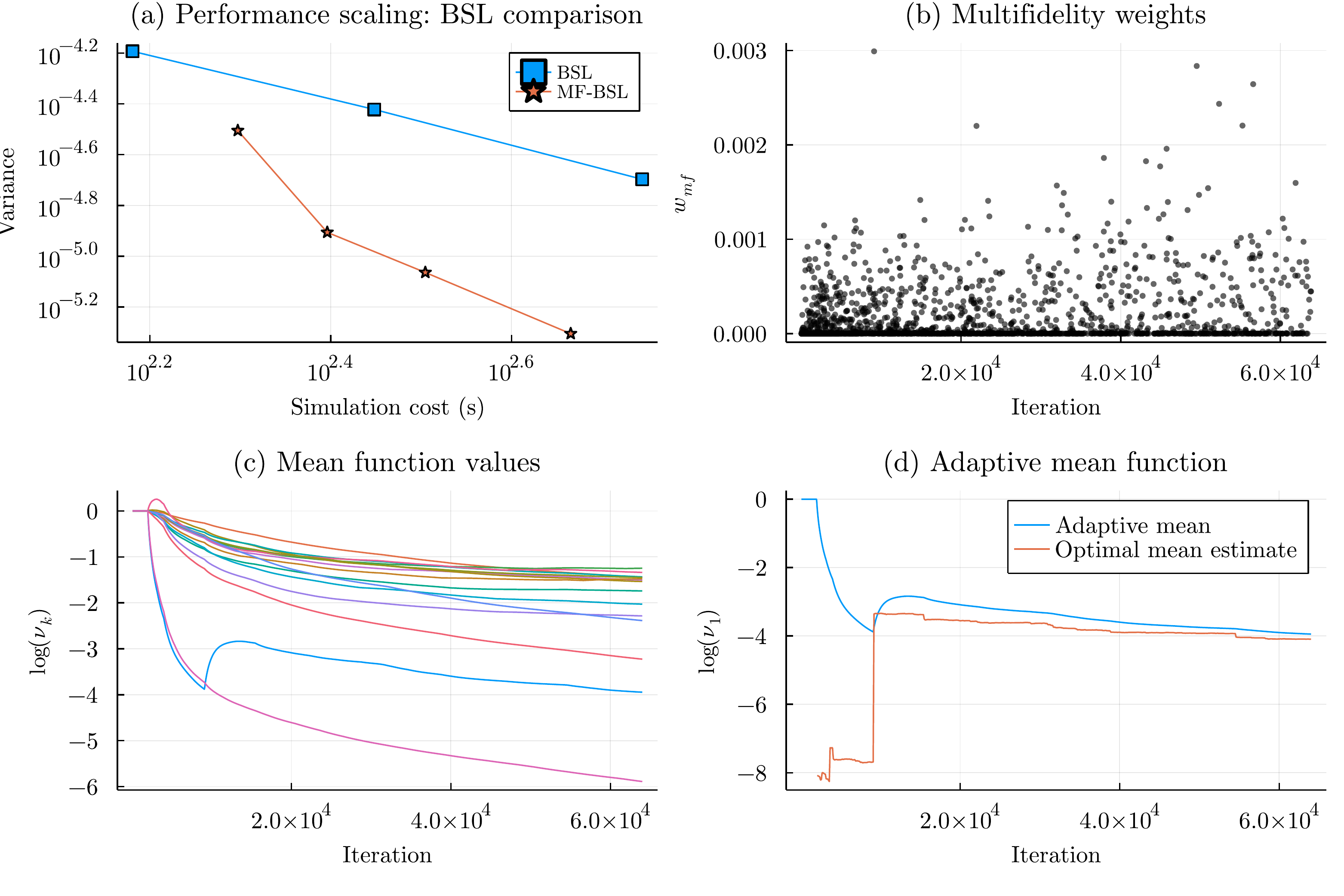}
\caption{%
	Multifidelity BSL.
	(a) Total simulation cost versus estimated variance of output estimate, $\hat G$, for three runs of \Cref{alg:is} (BSL) and four runs of \Cref{alg:mf2} (MF-BSL).
	(b) Values of the multifidelity weight, $w_\mf$, during the longest run of \Cref{alg:mf2}, for iterations where $\omega_\mf \neq \omega_\lo$, such that the low-fidelity likelihood-free weighting is corrected based on at least one high-fidelity simulation.
	(c) The evolution of the values of $\nu_k^{(i)}$ during the longest run of \Cref{alg:mf2}.
	(d) A comparison of the adaptive $\nu_1^{(i)}$ to the evolving best estimate of the optimal $\nu_1^\star$, given by \Cref{eq:nustar}, based on the Monte Carlo estimates in \Cref{eq:numc}.
}%
\label{fig:BSL}
\end{figure}

\Cref{fig:BSL} depicts the performance of multifidelity Bayesian synthetic likelihood (BSL) inference, where \Cref{alg:mf2} is applied with BSL likelihood-free weightings, $\omega_\lo$ and $\omega_\hi$.
As with MF-ABC, \Cref{fig:BSL}a shows that the MF-BSL generates improved performance over high-fidelity BSL inference, achieving lower variance estimates for a given computational budget.
We also note in \Cref{fig:BSL}a that the curve corresponding to MF-BSL has slope less than $-1$.
This is due to (a) the overhead cost of the initial burn-in period of \Cref{alg:mf2}, and also (b) the conservative convergence of $\nu^{(i)}$ to the optimum, as shown in \Cref{fig:BSL}c--d.
Both observations imply that earlier iterations are less efficiently produced than later iterations, meaning that larger samples show greater improvements than expected from the reciprocal relationship in \Cref{eq:perf_mf}.

Comparing \Cref{fig:BSL}b to \Cref{fig:ABC}b, we note that there are very few negative multifidelity weightings in MF-BSL, in comparison to MF-ABC.
We can conclude that the Bayesian synthetic likelihood, constructed using low-fidelity simulations, tends to underestimate the likelihood of the observed data compared to using high-fidelity simulations.
We note also in this comparison that the multifidelity likelihood-free weightings are on significantly different scales.

\section{Discussion}\label{s:end}
The characteristic computational burden of simulation-based, likelihood-free Bayesian inference methods is often a barrier to their successful implementation.
Multifidelity simulation techniques have previously been shown to improve the efficiency of likelihood-free inference in the context of ABC.
In this work, we have demonstrated that these techniques can be readily applied to general likelihood-free approaches.
Furthermore, we have introduced a computational methodology for automating the multifidelity approach, adaptively allocating simulation resources across different fidelities in order to ensure near-optimal efficiency gains from this technique.
As parameter space is explored, our methodology, given in \Cref{alg:mf2}, learns the relationships between simulation accuracy and simulation costs at the different fidelities, and adapts the requirement for high-fidelity simulation accordingly.

The multifidelity approach to likelihood-free inference is one of a number of strategies for speeding up inference, which include MCMC and SMC sampling techniques~\citep{Marjoram2003,Sisson2007,Toni2009} and methods for variance reduction such as multilevel estimation~\citep{Giles2015,Guha2017,Warne2018,Jasra2019}.
A key observation in the previous work of \citet{Prescott2021} and \citet{Warne2021} is that applying multifidelity techniques provides `orthogonal' improvements that combine synergistically with these other established approaches to improving efficiency.
Similarly, we envision that \Cref{alg:mf2} can be adapted into an SMC or multilevel algorithm with minimal difficulty, following the templates set by \citet{Prescott2021} and \citet{Warne2021}.

The multifidelity approach discussed in this work is a highly flexible generalisation of existing multifidelity techniques, which can be viewed as special cases of \Cref{alg:mf}.
In each of MF-ABC~\citep{Prescott2020,Prescott2021}, LZ-ABC~\citep{Prangle2016}, and DA-ABC~\citep{Everitt2021}, it is assumed that $\omega_\hi$ is an ABC likelihood-free weighting, which we relax in this work.
Furthermore, LZ-ABC and DA-ABC both use $\omega_\lo \equiv 0$, so that parameters are always rejected if no high-fidelity simulation is completed.
Clearly, we relax this assumption to allow for any low-fidelity likelihood-free weighting.
In all of MF-ABC, LZ-ABC and DA-ABC, the conditional distribution of $M$, given a parameter value $\theta$ and low-fidelity simulation output $\mathbf y_\lo$ is Bernoulli distributed, with mean $\mu(\theta, \mathbf y_\lo) \in (0,1]$.
In this work we change this distribution to Poisson, to ease analytical results, but any conditional distribution for $M$ can be used.
These adaptations are explored further in \Cref{s:unconstrained_optimisation}.

In the case of MF-ABC (as originally formulated by \citet{Prescott2020}) and DA-ABC \citep{Christen2005,Everitt2021}, the mean function, $\mu(\theta, y_\lo)$, depends on a single low-fidelity simulation and is assumed to be piecewise constant in the value of the indicator function $\mathbf 1(d(y_\lo, y_0) < \epsilon)$. 
LZ-ABC is more generic in its definition of $\mu = \mu(\phi(\theta, \mathbf y_\lo))$ to depend on the value of any \emph{decision statistic}, $\phi$.
In this work, we consider more general piecewise constant mean functions, $\mu_{\mathcal D}$, for heuristically derived partitions $\mathcal D$ of $(\theta, \mathbf y_\lo)$-space.
We observe that $(\theta, \mathbf y_\lo)$ may be of very high dimension; in the BSL example in \Cref{s:bsl}, having $K=100$ low-fidelity simulations $y_\lo \in \mathbb R^{10}$ means that the input to $\mu$ is of dimension $1003$.
In this situation, it may be tempting to seek a mean function that only depends on $\theta$.
However, we recall that the {optimal} mean function, $\mu^\star(\theta, \mathbf y_\lo)$, derived in \Cref{lem:mustar}, depends on the conditional expectation $\mathbf E((\omega_\hi - \omega_\lo)^2 \mid \theta, \mathbf y_\lo)$.
Thus, by ignoring $\mathbf y_\lo$, we would ignore the information about $\omega_\hi$ given by the evaluation of $\omega_\lo(\theta, \mathbf y_\lo)$.
Furthermore, the high dimension of the inputs to $\mu^\star$ suggest that this function is not necessarily well-approximated by a decision tree.
Future work may focus on methods to learn the optimal mean function directly without resorting to piecewise constant approximations~\citep{Levine2021}.
The key problem is ensuring the conservatism of any alternative estimate of $\mu^\star$, recalling that the variance of $w_\mf$ is inversely proportional to $\mu$.

In the example explored in \Cref{s:eg}, we considered the use of \Cref{alg:mf2} where $\omega_\hi$ and $\omega_\lo$ were first both ABC likelihood-free weightings, and then both BSL likelihood-free weightings.
In principle, this method should also allow for $\omega_\lo$ to be, for example, an ABC likelihood-free weighting based on a single low-fidelity simulation, and $\omega_\hi$ to be a BSL likelihood-free weighting based on $K>1$ high-fidelity simulations.
However, the success of the multifidelity method depends explicitly on the function $\eta(\theta, \mathbf y_\lo) = \mathbf E((\omega_\hi - \omega_\lo)^2 \mid \theta, \mathbf y_\lo)$ being sufficiently small, as quantified in \Cref{cor:muexist}.
If $\omega_\lo$ and $\omega_\hi$ are on different scales, as is likely when one is an ABC weighting and one a BSL weighting, then this function is not sufficiently small in general, and so the multifidelity approach fails.
We note, however, that we could instead consider the scaled low-fidelity weighting, $\tilde \omega_\lo = \gamma \mathbf \omega_\lo$, in place of $\omega_\lo$ in \Cref{alg:mf,alg:mf2} with no change to the target distribution.
Here, $\gamma$ is an additional parameter that can be tuned with $\mu$ when minimising the performance metric, $\mathcal J_\mf$; the optimal value of this parameter would need to be learned in parallel with the optimal mean function, $\mu$.
We defer this adaptation to future work.

Finally, this work follows \citet{Prescott2020,Prescott2021} in considering only a single low-fidelity model.
There is significant scope for further improvements by applying these approaches to suites of low-fidelity approximations~\citep{Gorodetsky2021}.
For example, exact stochastic simulations of biochemical networks, such as that simulated in \Cref{s:eg}, may also be approximated by tau-leaping~\citep{Gillespie2001,Warne2019}, where the time discretisation parameter $\tau$ tends to be chosen to trade off computational savings against accuracy: exactly the trade-off explored in this work.
Clearly, this parameter therefore has important consequences for the success of a multifidelity inference approach using such an approximation strategy.
More generally, a full exploration of the use of \emph{multiple} low-fidelity model approximations will be vital for the full potential of multifidelity likelihood-free inference to be realised.

\subsection*{Acknowledgements}
{\small\itshape
REB and TPP acknowledge funding for this work through the BBSRC/UKRI grant BB/R00816/1.
TPP is supported by the Alan Turing Institute and by Wave 1 of the UKRI Strategic Priorities Fund, under the ``Shocks and Resilience'' theme of the EPSRC/UKRI grant EP/W006022/1.
DJW thanks the Australian Mathematical Society for the Lift-off Fellowship, and acknowledges continued support from the Centre for Data Science at QUT and the ARC Centre of Excellence in Mathematical and Statistical Frontiers (ACEMS; CE140100049).
REB is supported by a Royal Society Wolfson Research Merit Award.
}

\newpage
\appendix
\section{Analytical results: Comparing performance}
\label{s:unconstrained_optimisation}

\subsection{\Cref{thm:performance}}
\begin{proof}
The leading order performance of each of \Cref{alg:is} and \Cref{alg:mf} is given in terms of increasing computational budget, $C_{\mathrm{tot}}$, in \Cref{eq:perf_hi} and \Cref{eq:perf_mf}, respectively.
For the performance of \Cref{alg:mf} to exceed that of \Cref{alg:is}, we compare the leading order coefficients from \Cref{eq:perf_hi,eq:perf_mf}, requiring
\begin{equation}
\frac{\mathbf E(C_\mf) \mathbf E(w_\mf^2 \Delta^2)}{\mathbf E(w_\mf)^2}
<
\frac{\mathbf E(C_\hi) \mathbf E(w_\hi^2 \Delta^2)}{\mathbf E(w_\hi)^2}.
\label{eq:compare_coefficients}
\end{equation}
We note that $\mathbf E(w_\mf \mid \theta) = \pi(\theta) L_\mf(\theta) / q(\theta)$ and $\mathbf E(w_\hi \mid \theta) = \pi(\theta) L_\hi(\theta) / q(\theta)$.
Since $L_\mf = L_\hi$, as shown in \Cref{prop:mf_accuracy}, the denominators in \Cref{eq:compare_coefficients} are therefore equal.
Thus,
\[
\mathcal J_\mf = \mathbf E(C_\mf) \mathbf E(w_\mf^2 \Delta^2)
<
\mathbf E(C_\hi) \mathbf E(w_\hi^2 \Delta^2) = \mathcal J_\hi,
\]
is the condition for \Cref{alg:mf} to outperform \Cref{alg:is}.

Taking the right-hand side of this inequality first, clearly the expected simulation time is $\mathbf E(C_\hi) = \bar c_\hi$, for the constant $\bar c_\hi$ defined in \Cref{eq:c_hi}.
Similarly, we can write
\[
\mathbf E(w_\hi^2 \Delta^2) = \int \left( \frac{\pi(\theta)}{q(\theta)} \Delta(\theta) \right)^2 \left[ \int \omega_\hi(\theta, \mathbf y_\hi)^2  f_\hi(\mathbf y_\hi \mid \theta) \mathrm d \mathbf y_\hi \right] q(\theta) \mathrm d\theta = V_\hi,
\]
as given in \Cref{eq:V_hi}.
Thus, $\mathcal J_\hi = \bar c_\hi V_\hi$.

For the left-hand side of the performance inequality, we take each expectation in $\mathcal J_\mf$ in turn.
We first note that the expected iteration cost of \Cref{alg:mf}, $\mathbf E(C_\mf)$, is the sum of the expected cost of a single low-fidelity simulation, and the expected cost of $M$ high-fidelity simulations.
By definition, the expected cost of a single low-fidelity simulation $\mathbf y_\lo \sim f_\lo(\cdot \mid \theta)$ across $\theta \sim q(\cdot)$ is given by $\bar c_\lo$.
Thus the remaining cost, $\mathbf E(\delta C_\mf) = \mathbf E(C_\mf) - \bar c_\lo$, is the expected cost of $M$ high-fidelity simulations.
Conditioning on $\theta$, $\mathbf y_\lo$ and $M=m$, the expected remaining cost is, by definition,
\[
\mathbf E(\delta C_\mf \mid \theta, \mathbf y_\lo, M=m) = m c_\hi(\theta, \mathbf y_\lo).
\]
Taking expectations over the conditional distribution $M \sim \mathrm{Poi}(\mu(\theta, \mathbf y_\lo))$, we have 
\[
\mathbf E(\delta C_\mf \mid \theta, \mathbf y_\lo) = \mu(\theta, \mathbf y_\lo) c_\hi(\theta, \mathbf y_\lo).
\]
Finally, integrating this expression over the density $\rho$ in \Cref{eq:rho} gives the first factor of \Cref{eq:J_mf}.

It remains to show that
\[
\mathbf E(w_\mf^2 \Delta^2) = V_\mf + \iint \Delta_q(\theta)^2 \frac{\eta(\theta, \mathbf y_\lo)}{\mu(\theta, \mathbf y_\lo)} ~\rho(\theta, \mathbf y_\lo) \mathrm d\theta \mathrm d \mathbf y_\lo.
\]
We first condition on $\theta$, $\mathbf y_\lo$ and $M=m$, to write
\begin{align*}
\mathbf E(w_\mf^2 \Delta^2 \mid \theta, \mathbf y_\lo, m) 
&= \Delta_q^2 \mathbf E(\omega_\mf^2 \mid \theta, \mathbf y_\lo, m) \\
&= \Delta_q^2 \left[
	\omega_\lo^2 
	+ \frac{2}{\mu} \omega_\lo \mathbf E \left( D_m  \mid \theta, \mathbf y_\lo \right)
	+ \frac{1}{\mu^2} \mathbf E \left( D_m^2 \mid \theta, \mathbf y_\lo \right)
\right],
\end{align*}
for the random variable $D_m = \sum_{i=1}^m \left( \omega_{\hi,i} - \omega_\lo \right)$.
It is straightforward to show that 
\begin{align*}
\mathbf E(D_m \mid \theta, \mathbf y_\lo) &= m \mathbf E(\omega_\hi \mid \theta, \mathbf y_\lo) - m\omega_\lo, \\
\mathbf E(D_m^2 \mid \theta, \mathbf y_\lo) &= m \mathbf E \left( (\omega_\hi - \omega_\lo)^2 \mid \theta, \mathbf y_\lo \right) + (m^2 - m) \mathbf E(\omega_\hi - \omega_\lo \mid \theta, \mathbf y_\lo)^2,
\end{align*}
where we exploit the conditional independence of the high-fidelity simulations $\mathbf y_{\hi,i}$ and $\mathbf y_{\hi,j}$, for $i \neq j$.
On substitution of these conditional expectations, we then rearrange to write
\[
\mathbf E(w_\mf^2 \Delta^2 \mid \theta, \mathbf y_\lo, m)
= \Delta_q^2 \left[ 
	\left( 1 - \frac{2m}{\mu} \right) \omega_\lo^2 
	+ \frac{2m}{\mu} \omega_\lo \lambda_\hi
	+ \frac{m}{\mu^2} \mathrm{Var}(\omega_\hi - \omega_\lo \mid \theta, \mathbf y_\lo)
	+ \left( \frac{m(\lambda_\hi-\omega_\lo)}{\mu} \right)^2
\right],
\]
where we write the conditional expectation $\lambda_\hi(\theta, \mathbf y_\lo) = \mathbf E(\omega_\hi \mid \theta, \mathbf y_\lo)$.
At this point we can take expectations over $M$ and rearrange to give
\begin{align}
\mathbf E(w_\mf^2 \Delta^2 \mid \theta, \mathbf y_\lo)
&= 
\Delta_q^2 \left[
	2 \omega_\lo \lambda_\hi 
	- \omega_\lo^2
	+ \frac{1}{\mu} \mathrm{Var}(\omega_\hi - \omega_\lo \mid \theta, \mathbf y_\lo)
	+ \frac{(\mathrm{Var}(M \mid \theta, \mathbf y_\lo) + \mu^2) (\lambda_\hi - \omega_\lo)^2}{\mu^2}
\right] \nonumber \\
&= \Delta_q^2 \left[ \lambda_\hi^2 + \frac{1}{\mu} \left( \mathrm{Var}(\omega_\hi - \omega_\lo \mid \theta, \mathbf y_\lo) + \frac{\mathrm{Var}(M \mid \theta, \mathbf y_\lo)}{\mu} \left( \mathbf E(\omega_\hi - \omega_\lo \mid \theta, \mathbf y_\lo) \right)^2 \right) \right]. \label{eq:varstage}
\end{align}
Here, we can use the assumption that $M$ conditioned on $\theta$ and $\mathbf y_\lo$ is Poisson distributed, noting that the statement of \Cref{thm:performance} can be adapted for other conditional distributions of $M$ with different conditional variance functions.
Under the Poisson assumption, we can substitute $\mathrm{Var}(M \mid \theta, \mathbf y_\lo) = \mu(\theta, \mathbf y_\lo)$ to give
\[
\mathbf E(w_\mf^2 \Delta^2 \mid \theta, \mathbf y_\lo)
=
\Delta_q^2 \left[ \lambda_\hi^2 + \frac{\mathbf E\left(\left(\omega_\hi - \omega_\lo \right)^2 \mid \theta, \mathbf y_\lo \right)}{\mu(\theta, \mathbf y_\lo)} \right].
\]
Finally, we take expectations with respect to the probability density $\rho$ in \Cref{eq:rho}, and the product in \Cref{eq:J_mf} follows.
\end{proof}

\subsubsection{Alternative conditional distributions for $M$}

The proof above derives the performance measure $\mathcal J_\mf$ given in \Cref{eq:J_mf}, under the assumption that the conditional distribution of $M$, given $\theta$ and $\mathbf y_\lo$, is Poisson with mean $\mu(\theta, \mathbf y_\lo)$.
The following corollaries adapt the expression for $\mathcal J_\mf$ in the case of alternative conditional distributions for $M$.
We first define the MSE,
\[
E_\mf = \iint \left[ \lambda_\hi(\theta, \mathbf y_\lo) - \omega_\lo(\theta, \mathbf y_\lo) \right]^2 ~\rho(\theta, \mathbf y_\lo) \mathrm d\theta \mathrm d \mathbf y_\lo,
\]
between $\omega_\lo(\theta, \mathbf y_\lo)$ and $\lambda_\hi(\theta, \mathbf y_\lo) = \mathbf E(\omega_\hi \mid \theta, \mathbf y_\lo)$.

\begin{cor}
If $M \sim \mathrm{Bin}(M_{\max}, p(\theta, \mathbf y_\lo))$ is binomially distributed with maximum value $M_{\max}$ and mean $\mu(\theta, \mathbf y_\lo)$, where $p=\mu/M_{\max}$, then 
\begin{align}
\mathcal J_\mf[\mu] &= 
\left( \bar c_\lo + \iint \mu(\theta, \mathbf y_\lo) c_\hi(\theta, \mathbf y_\lo) ~\rho(\theta, \mathbf y_\lo) \mathrm d\theta \mathrm d\mathbf y_\lo \right) \nonumber \\
&\quad{} \times
\left( V_\mf - \frac{E_\mf}{M_{\max}} + \iint \Delta_q(\theta)^2 \frac{\eta(\theta, \mathbf y_\lo)}{\mu(\theta, \mathbf y_\lo)} ~\rho(\theta, \mathbf y_\lo) \mathrm d\theta \mathrm d \mathbf y_\lo \right). \label{eq:J_mf_Bin}
\end{align}
\end{cor}
\begin{proof}
We substitute $\mathrm{Var}(M \mid \theta, \mathbf y_\lo) = \mu \left( 1 - \mu/M_{\max} \right)$ into \Cref{eq:varstage}, and the result follows.
\end{proof}

We note in the result above that for $\mu$ to be the conditional mean of $M \sim \mathrm{Bin}(M_{\max}, p(\theta, \mathbf y_\lo))$, we must constrain the values of $\mu$ such that $\mu(\theta, \mathbf y_\lo) \in (0, M_{\max}]$.
This constraint alters the derivation of the optimal $\mu^\star$, in the case of a binomial conditional distribution with fixed $M_{\max}$.

\begin{cor}
If $M \sim \mathrm{Geo}(p(\theta, \mathbf y_\lo))$ is geometrically distributed on the non-negative integers, with mean $\mu(\theta, \mathbf y_\lo)$, where $p = 1/(1+\mu)$, then
\begin{subequations}
\begin{align}
\mathcal J_\mf[\mu] &= 
\left( \bar c_\lo + \iint \mu(\theta, \mathbf y_\lo) c_\hi(\theta, \mathbf y_\lo) ~\rho(\theta, \mathbf y_\lo) \mathrm d\theta \mathrm d\mathbf y_\lo \right) \nonumber \\
&\quad{} \times
\left( V_\mf + E_\mf + \iint \Delta_q(\theta)^2 \frac{\eta(\theta, \mathbf y_\lo)}{\mu(\theta, \mathbf y_\lo)} ~\rho(\theta, \mathbf y_\lo) \mathrm d\theta \mathrm d \mathbf y_\lo \right). \label{eq:J_mf_Geo}
\end{align}
\end{subequations}
\end{cor}
\begin{proof}
We substitute $\mathrm{Var}(M \mid \theta, \mathbf y_\lo) = \mu \left( 1 + \mu \right)$ into \Cref{eq:varstage}, and the result follows.
\end{proof}

\subsection{\Cref{lem:mustar}}
We return to the assumption that $M$ is conditionally Poisson distributed, given $\theta$ and $\mathbf y_\lo$.
\begin{proof}
We write the functional $\mathcal J_\mf[\mu] = \mathcal C[\mu] \mathcal V[\mu]$ in \Cref{eq:J_mf} as the product of functionals,
\begin{subequations}
\label{eq:CV}
\begin{align}
\mathcal C[\mu] &= \bar c_\lo + \iint \mu c_\hi ~\rho \mathrm d\theta \mathrm d\mathbf y_\lo, \\
\mathcal V[\mu] &= V_\mf + \iint \Delta_q^2 \frac{\eta}{\mu} ~\rho\mathrm d\theta \mathrm d \mathbf y_\lo.
\end{align}
\end{subequations}
Standard `product rule' results from calculus of variations allows us to write the functional derivative of $\mathcal J_\mf$ with respect to $\mu$ as
\begin{align*}
\frac{\delta \mathcal J_\mf}{\delta \mu}
&=
\mathcal V[\mu] \frac{\delta \mathcal C}{\delta \mu} + \mathcal C[\mu] \frac{\delta \mathcal V}{\delta \mu} \\
&= \mathcal V[\mu] c_\hi \rho - \mathcal C[\mu] \frac{\Delta_q^2 \eta \rho}{\mu^2}.
\end{align*}
Setting this functional derivative to zero, the optimal function, $\mu^\star$, satisfies
\begin{equation}
\mu^\star(\theta, \mathbf y_\lo)^2 = \frac{\mathcal C[\mu^\star]}{\mathcal V[\mu^\star]} \frac{ \Delta_q(\theta)^2 \eta(\theta, \mathbf y_\lo)}{c_\hi(\theta, \mathbf y_\lo)}.
\label{eq:mustar_almost}
\end{equation}
The result in \Cref{eq:mustar} follows on showing that $C[\mu^\star] / V[\mu^\star] = \bar c_\lo / V_\mf$.

On substituting \Cref{eq:mustar_almost} into \Cref{eq:CV} we find
\begin{align*}
\mathcal C[\mu^\star] &= \bar c_\lo + \sqrt{\frac{\mathcal C[\mu^\star]}{\mathcal V[\mu^\star]}} \iint \sqrt{\Delta_q^2 \eta c_\hi} ~\rho\mathrm d\theta \mathrm d\mathbf y_\lo, \\
\mathcal V[\mu^\star] &= V_\mf + \sqrt{\frac{\mathcal V[\mu^\star]}{\mathcal C[\mu^\star]}} \iint \sqrt{\Delta_q^2 \eta c_\hi} ~\rho\mathrm d\theta \mathrm d\mathbf y_\lo,
\end{align*}
from which it follows that
\[
\sqrt{\frac{\mathcal V[\mu^\star]}{\mathcal C[\mu^\star]}} \bar c_\lo
=
\sqrt{\frac{\mathcal C[\mu^\star]}{\mathcal V[\mu^\star]}} V_\mf
=
\sqrt{\mathcal C[\mu^\star] \mathcal V[\mu^\star]} - \iint \sqrt{\Delta_q^2 \eta c_\hi} ~\rho\mathrm d\theta \mathrm d\mathbf y_\lo.
\]
Multiplying this equation by $\sqrt{\mathcal C[\mu^\star] \mathcal V[\mu^\star]}$, we have $\mathcal V[\mu^\star] \bar c_\lo = \mathcal C[\mu^\star] V_\mf$, and thus \Cref{eq:mustar} follows from \Cref{eq:mustar_almost}.
\end{proof}

\subsection{\Cref{cor:muexist}}
\label{s:muexist}
\begin{proof}
On substituting \Cref{eq:mustar} into \Cref{eq:CV}, we find that the condition $\mathcal J_\mf^\star = \mathcal J_\mf[\mu^\star] < \mathcal J_\hi = \bar c_\hi V_\hi$ is equivalent to
\[
\left( \sqrt{\bar c_\lo V_\mf} + \iint \sqrt{\Delta_q(\theta)^2 \eta(\theta, \mathbf y_\lo) c_\hi(\theta, \mathbf y_\lo)} ~\rho(\theta, \mathbf y_\lo) \mathrm d\theta \mathrm d\mathbf y_\lo \right)^2 < \bar c_\hi V_\hi.
\]
A simple rearrangement of this inequality gives the inequality in \Cref{eq:muexist}.
\end{proof}

To interpret the condition
\[
\sqrt{\frac{\bar c_\lo}{\bar c_\hi} \frac{V_\mf}{V_\hi}} + \iint \sqrt{\frac{\Delta_q(\theta)^2 \eta(\theta, \mathbf y_\lo)}{V_\hi}} \sqrt{\frac{c_\hi(\theta, \mathbf y_\lo)}{\bar c_\hi}} ~\rho(\theta, \mathbf y_\lo) \mathrm d\theta \mathrm d\mathbf y_\lo
< 1
\]
in \Cref{eq:muexist}, we note that the first term is determined by (a) our assumption of a significant reduction in simulation burden of the low-fidelity model over the high-fidelity model, $\bar c_\lo < \bar c_\hi$, and (b) the ratio of the two integrals,
\[
\frac{V_\mf}{V_\hi} = \frac{
	\int \Delta_q(\theta)^2 \mathbf E( \mathbf E(\omega_\hi \mid \theta, \mathbf y_\lo)^2 \mid \theta) q(\theta) ~\mathrm d\theta
}{
	\int \Delta_q(\theta)^2 \mathbf E(\omega_\hi^2 \mid \theta) q(\theta) ~\mathrm d\theta
}.
\]
Exploiting the law of total variance, we note that
\begin{align*}
\mathbf E(\omega_\hi^2 \mid \theta) &= \mathrm{Var}(\omega_\hi \mid \theta) + L_\hi(\theta)^2, \\
\mathbf E(\mathbf E(\omega_\hi \mid \theta, \mathbf y_\lo)^2 \mid \theta) &= \mathrm{Var}( \mathbf E(\omega_\hi \mid \theta, \mathbf y_\lo) \mid \theta) + L_\hi(\theta)^2 \\
&= \mathbf E(\omega_\hi^2 \mid \theta) - \mathbf E( \mathrm{Var}(\omega_\hi \mid \theta, \mathbf y_\lo) \mid \theta).
\end{align*}
These equalities imply that
\[
\mathbf E(\omega_\hi \mid \theta)^2
\leq
\mathbf E(\mathbf E(\omega_\hi \mid \theta, \mathbf y_\lo)^2 \mid \theta)
\leq
\mathbf E(\omega_\hi^2 \mid \theta),
\]
where the lower bound is achieved for $\mathbf y_\hi$ independent of $\mathbf y_\lo$,
while the upper bound would be achieved if $\mathbf y_\hi$ were a deterministic function of $\mathbf y_\lo$.
In particular, $V_\mf / V_\hi \leq 1$, and so the first term of \Cref{eq:muexist} is small whenever the low-fidelity model provides significant computational savings versus the high-fidelity model.

The second term in \Cref{eq:muexist} quantifies the detriment to the performance of \Cref{alg:mf} that arises from the inaccuracy of $\omega_\lo$ as an estimate of $\omega_\hi$.
The function $\eta(\theta, \mathbf y_\lo) = \mathbf E((\omega_\hi - \omega_\lo)^2 \mid \theta, \mathbf y_\lo)$ is integrated across the density $\rho$, weighted by the relative computational cost of the high-fidelity simulation, $c_\hi(\theta, \mathbf y_\lo) / \bar c_\hi$, and by the contribution of $G(\theta)$ to the variance of the estimated posterior expectation of $G$.
We can conclude that the multifidelity approach requires that the low-fidelity model is accurate in the regions of parameter space where high-fidelity simulations are particularly expensive.

To summarise: if 
(a) the ratio between average low-fidelity simulation costs and high-fidelity simulation costs is suitably small, and
(b) the average disagreement between likelihood-free weightings, as measured by $\eta$, is suitably small,
then \Cref{eq:muexist} will be satisfied and thus a mean function, $\mu^\star$, exists such that \Cref{alg:mf} is more efficient than \Cref{alg:is}.

\newpage
\section{Mean functions}
\label{s:mean_eg}

\begin{algorithm}[h!]
\caption{Piecewise constant mean function $\mu_{\mathcal D}(\theta, y_\lo; \nu)$ used in MF-ABC \Cref{alg:mf2}, depicted in \Cref{fig:ABC}c, at final iteration.}
\label{alg:mu_ABC}
\include{mu_ABC}
\end{algorithm}



\newpage
\bibliography{refs}
\end{document}

%% file: mu_ABC.tex
\begin{algorithmic}
\Require $\theta=(k_1, k_{-1}, k_2)$; $y_{\lo} = (y_1, y_2, \dots, y_{10})$.
\If{$ y_{7} \leq 13.867 $}
\State \Return $\nu_{1} = 0.084$.
\Else
\If{$ k_2 \leq 1.14 $
}\If{$ y_{8} \leq 16.219 $}
\State \Return $\nu_{2} = 0.616$.
\Else
\If{$ k_2 \leq 0.88 $
}\State \Return $\nu_{3} = 0.08$.
\Else
\If{$ y_{10} \leq 26.208 $}
\If{$ k_{1} \leq 91.265 $
}\State \Return $\nu_{4} = 0.313$.
\Else
\State \Return $\nu_{5} = 0.76$.
\EndIf
\Else
\If{$ y_{7} \leq 15.151 $}
\State \Return $\nu_{6} = 0.761$.
\Else
\If{$ y_{1} \leq 3.371 $}
\If{$ y_{5} \leq 11.264 $}
\State \Return $\nu_{7} = 0.688$.
\Else
\State \Return $\nu_{8} = 0.467$.
\EndIf
\Else
\State \Return $\nu_{9} = 0.797$.
\EndIf
\EndIf
\EndIf
\EndIf
\EndIf
\Else
\If{$ y_{10} \leq 26.136 $}
\If{$ y_{7} \leq 14.17 $}
\State \Return $\nu_{10} = 0.929$.
\Else
\State \Return $\nu_{11} = 0.828$.
\EndIf
\Else
\State \Return $\nu_{12} = 1.439$.
\EndIf
\EndIf
\EndIf
\end{algorithmic}